\documentclass[preprint,12pt]{elsarticle}
\usepackage{amsmath} 
\usepackage{amssymb} 
\usepackage{xcolor}
\usepackage{subfig} 
\usepackage{multirow}

\newcommand{\matr}[1]{\boldsymbol{#1}}

\newcommand{\vect}[1]{\boldsymbol{#1}}
\newcommand{\eqdef}{\stackrel{def}{=}}
\newcommand{\E}[1]{\mathop{\mbox{$\mathbb{E}$}}\{#1\}} 
\newcommand{\I}{\mathbb I}
\newcommand{\RR}{\mathbb{R}}
\newtheorem{theorem}{Theorem}

\newtheorem{assumption}[theorem]{Assumption}
\newtheorem{proposition}[theorem]{Proposition}
\newtheorem{lemma}[theorem]{Lemma}
\newtheorem{remark}[theorem]{Remark}
\newenvironment{proof}{ \vspace{\smallskipamount}\par {\it Proof.}~}{\hfill $\Box$ \vspace{\medskipamount}\par }

\newcommand{\FIN}{\color{black}}

\begin{document}\sloppy

\begin{frontmatter}
\title{A Normality Test for Multivariate Dependent Samples}
\author{Sara ElBouch and Olivier Michel and Pierre Comon}
\affiliation{organization={Univ. Grenoble Alpes, CNRS, Grenoble INP, GIPSA-Lab},
            addressline={Grenoble Campus, BP.46}, 
            city={Grenoble},
            postcode={38000}, 
            country={France}}

\begin{abstract}

Most normality tests in the literature are performed for scalar and independent samples. Thus, they become unreliable when applied to colored processes,  hampering their use in realistic scenarios.
We focus on Mardia's multivariate kurtosis, derive closed-form expressions of its asymptotic distribution for statistically dependent samples, under the null hypothesis of normality and a mixing condition. The calculus is long and tedious but the final result is simple and is implemented with a low computational burden. The proposed expression of the test power exhibits good properties on various scenarios; this is illustrated by computer experiments by means of copulas. 
\end{abstract}



\end{frontmatter}

\section{Introduction}\label{sec-intro}

The interest in techniques involving higher order statistics has grown considerably during the past decades \citep{NikiP93,Hayk00,CichA02,ComoJ10}. Actually, first and second order statistics allow an exhaustive characterization of Gaussian processes and linear systems.
Despite the practical importance of the Gaussian distribution, thanks to the central limit theorem, and the prevalence of linear dynamical systems  in small fluctuations models, many situations do not resort to these assumptions. As a consequence, detecting departure from Gaussianity arose as a means to detect and characterize non linear behavior, detection of changes in dynamical regimes \citep{BassN93}, etc. Higher-Order Statistics (HOS) were also shown to carry valuable information for blind identification problems, source separation and in measuring information theoretic quantities \citep{ComoJ10}, to name a few applications.  

 The present growth of interest in sensor networks and our ability to simultaneously record time series representing the fluctuations of numerous physical quantities, naturally leads to consider  $p$-dimensional processes. Surprisingly enough, normality tests for such $p$-dimensional stochastic processes were not so much investigated. To be more precise, very few results concern both the multivariate nature of the time series and the fact that the $p$-dimensional time samples cannot in general be considered as being i.i.d. This will be referred to as the \emph{non independent} identically distributed (n.i.d.)  property. The difficulty in testing the Gaussian nature of $p$-dimensional stochastic processes arises from the necessity for the test to tackle both the joint Gaussianity of the $p$ components and the time dependence of successive $p$-dimensional samples. To make this framework clear, the following notation is introduced: Let $\vect{x}(n)=[x_1(n),\dots,x_p(n)]^T$ be a real $p$-variate stochastic process, of which a sample of finite size, $N$, is observed. Hereafter, the stochastic processes under study  will be assumed stationary and zero-mean with covariance matrices for delay $\tau$ defined by 
\begin{equation}
    \matr{S}(\tau)=\E{\vect{x}(n)\vect{x}(n-\tau)^T}.
\end{equation}
Let $S_{ab}(\tau) $ denote the entries of matrix $\matr{S}(\tau)$, $(a,b) \in \{1,\ldots,p\}$. 
The n.i.d. nature of the time samples corresponds to have $S_{ab}(\tau) \neq 0$ in general for $\tau\neq0$. Thus, the process $X(n)$ enjoys both \emph{spatial} and \emph{temporal} dependancies (temporal refers to dependence w.r.t. $\tau$ and spatial  w.r.t. to  $(a,b) \in \{1,\ldots, p\}$). 
The normality test \emph{without alternative} can be formulated as follows:
\begin{quote}
\textbf{Problem P1:}
Given a finite sample  of size $N$, $\matr{X}\eqdef\{\vect{x}(1),\dots,\vect{x}(N)\}$:
\begin{equation}\label{eq-testGeneral}
    H: \matr{X} ~ \text{is Gaussian} \quad versus \quad \bar{H}
\end{equation}
where variables $\vect{x}(n)\in\RR^p$ are identically distributed, but not statistically independent.
\end{quote}
Solving this problem requires (i)~to define a test variable, and (ii)~to determine its asymptotic distribution (often itself normal) in order to assess the power of the test, that is, the probability to decide $H$ whereas $H$ is true.

For the scalar case ($p=1$), since the so-called Chi-squared test proposed by Fisher and improved in \citep{Moor71:as}, the most popular test is probably the omnibus test based on skewness and kurtosis  \citep{BowmS75:biom}.
The omnibus test combines estimated skewness and kurtosis weighted by the inverse of their respective asymptotic variance, evaluated under the assumption that the samples are Gaussian i.i.d.; see also \citep{Mard74:sank}, \citep{KotzJ82}. 
The asymptotic distribution of the test is $\chi_2^2$ when samples are i.i.d normal. However, as pointed out by \citep{Moor82:as}, the Chi-square test is very sensitive to the dependence between samples; the process color yields a loss in apparent normality \citep{Gass75:biom}. Actually, most of the tests proposed in the literature assume that observations are i.i.d., see \citep{ShapWC68:asaj} or \citep{PearAB77:biom}. This is also true for multivariate tests \citep{Mard70:biom,AndrGW73:ma3}; see the survey of \citep{Henz02:statp}. 
Only very few authors address the case of n.i.d samples, or so-called colored processes. One can mention Hinich's bispectrum-based linearity test \citep{Hini82:jtsa},  or Brillinger's trispectrum  \citep{Bril81}.  These multispectra (Fourier transform of 3rd and 4th order cumulant multicorrelations evaluated under stationarity assumption) induce in general an important computational load and have important estimation variance  even for large values of $N$. 
An appealing alternative was proposed in \citep{MoulCC92:ssap} where non linear transforms of the samples allow to go  beyond   monomials of degree 3 or 4. For instance, some tests are based on the characteristic function \citep{Epps87:as,MoulCC92:ssap} and others on entropy \citep{SteiZ92:ieeeIT}.
These tests remain however complex to implement in practice and may hardly be  executable in real time on a light processor when samples are colored (\emph{i.e.} statistically time dependant). 

\smallskip

\textbf{Contribution.}
Taking an opposite direction, the purpose of this paper is to propose a normality test that is simple to implement, even for colored (time correlated) $p$-dimensional processes, eventually at the expense of quite complicated and lengthy calculus to derive  the exact form of the test.  
For this reason, we shall focus on the multivariate kurtosis proposed by Mardia in \citep{Mard70:biom}, and derive its mean and variance when samples are assumed to be statistically n.i.d. 
Although the results presented apply easily in practice to a very wide class of applications, the process is assumed  to satisfy statistical \emph{mixing} properties. Within this framework,  a general procedure is proposed to compute the asymptotic mean and variance of Mardia's multivariate kurtosis when applied to \emph{colored} processes, and is shown to converge in $O(1/\sqrt{N})$,  thus fully characterizing the asymptotic normal distribution of the test statistic. 
The complete derivation is given for $p=2$, which allows to test joint normality of arbitrary 2-dimensional projections, thus generalizing the tests proposed in \citep{Mard70:biom,MalkFA73:jasa,NietCG14:csda} based on 1D projections and i.i.d samples. The benefits of using 2D is clear, as the resulting tests are subsequently shown to outperform 1D projection-based tests, via computer experiments. The importance of joint normality and the performance of our test is illustrated on \emph{n.i.d.} copulas, \emph{i.e.} with colored Gaussian marginals. Additionally, the particular case where $p$-dimensional observations are constructed by time embedding, thus mixing time and space dependencies, is developed.

\medskip

This article is organized as follows. Section \ref{sec-multiKurt} contains the definition of the test statistic.  
The tools necessary to conduct the calculations are introduced in Section \ref{sec-tools}.  
The moments involved in the derivation of both the mean and variance of the test statistic are given in Sections \ref{sec-genexp-mean}-\ref{sec-genexp-variance},  for arbitrary dimension $p$. 
Their expression in closed form for $p=1$ and $p=2$, and for the case where the multivariate process arises from a time embedding,  are given in Sections  \ref{sec-scalexps}-\ref{sec-embed-exps}.  
Section \ref{sec-comp-exp} reports some computer experiments.   The expressions of  moments and  details of calculation are deferred to  appendices in Section \ref{sec:appendix}.

\section{   Mardia's   Multivariate kurtosis}\label{sec-multiKurt}

The test proposed   by Mardia   in \citep{Mard70:biom} takes the form:
\begin{equation}\label{eq:popkurtosis}
    \beta_p = \E{(\vect{x}^{T}\matr{S}^{-1}\vect{x})^{2}}. 
\end{equation}
For $\vect{x}\sim\mathcal{N}_p(0,\matr{S})$, one can show that $\beta_p=p(p+2)$. Its sample counterpart for a sample of size $N$ is:
\begin{equation}\label{eq:samplekurtosis}
    B_p(N) = \frac{1}{N} \sum_{n=1}^{N}(\vect{x}(n)^T\matr{S}^{-1}\vect{x}(n))^2
\end{equation}
  It is worth noticing that $\matr S$ being the exact covariance matrix, all random realizations involved in the latter equation are standardized (remind that we assume zero-mean processes). Thus, the    advantage of this test variable is that it is invariant with respect to linear transformations, i.e., $\vect{y}=\matr{A}\vect{x}$. In practice, the covariance matrix $\matr{S}$ is unknown and is replaced by its sample estimate, $\hat{\matr{S}}$, so that we end up with the following test variable:
\begin{equation}\label{estimKurt-eq}
     \hat{B}_p(N) = \frac{1}{N} \sum_{n=1}^{N}\big(\vect{x}(n)^T\hat{\matr{S}}^{-1}\vect{x}(n)\big)^2
\end{equation}
with
$$
\hat{\matr{S}}=\frac{1}{N}\sum_{k=1}^{N}\vect{x}(k)\vect{x}(k)^{T}.
$$
The multivariate normality test can be formulated in terms of the multivariate Kurtosis: the variable $\vect{x}$ is  said to be normal if $|\hat{B}_p(N)-\E{\hat{B}_p(N)| H_0}|\le \eta$, where $\eta$ is a threshold to be determined as a function of the power of the test. 
The fact that $\hat{B}_p(N)$ is a good estimate of $\beta_p$ or not is relevant; what is important is to have a sufficiently accurate estimation of the power of the test. 
In order to do that, we need to assess the mean and variance of $\hat{B}_p(N)$ under $H$. 
Under the assumption that $\vect{x}(n)$ are i.i.d. realizations of variable $\vect{x}$, the mean and variance of $\hat{B}_p(N)$ have been calculated:

\begin{theorem}\citep{Mard70:biom}\label{mk}
Let $\{\vect{x}(n)\}_{1 \leq n \leq N}$ be \textit{i.i.d.} of dimension $p$. Then under the null hypothesis $H_0$, $\hat{B}_p(N)$ is asymptotically normal, with mean $p(p+2)\frac{N-1}{N+1}$ and variance $\frac{8p(p+2)}{N}+o(\frac{1}{N})$.
\end{theorem}

  Note that the result above makes use of Landau notation $o(\frac{1}{N})$, to precise that the absolute approximation error is dominated by  $\frac{1}{N}$. Landau notation $O(h(N))$ will also be often used, when the absolute approximation error will be of the order of $h(N)$. Both will be extensively used in the rest of this paper.

Our purpose is now to state a similar theorem when $\{\vect{x}(n)\}$ are not independent. Since this involves heavy calculations, we need to introduce some tools to make them possible.

\section{Statistical and combinatorial tools}\label{sec-tools}
  In this section, partial useful results are established. Each is associated with a lemma, and represents a step towards the derivation of the exact expression of the statistics  of $\widehat{B}_p(N)$ defined in (\ref{estimKurt-eq}) for multivariate \emph{colored} processes: \\
\indent -Lemma \ref{lemma1}  proves that $\matr{\Delta} = \hat{\matr{S}}-\matr{S}$ varies as $O(1/\sqrt{N})$. \\
\indent -Lemma \ref{lemma2} uses the preceding result in order to express the sample precision matrix $\hat{G} = (\matr{S}+\matr{\Delta})^{-1}$ as a function of  the exact precision matrix $\matr{G}$ and of the approximation matrix $\matr{\Delta}$, up to order $O(\| \matr{\Delta} \|^3)$.\\
\indent -Finally, lemma \ref{lem:DLBp} allows to derive the approximate  expression of $\hat{B}_p(N)$ to order $N^{-1}$ in $o(1/N)$.

\smallskip

From now on, we assume the following condition 
upon $\vect{x}(n)$, necessary to relax the i.i.d. property while maintaining convergence of various terms: 
\begin{assumption}[Mixing]\label{mixing-ass}
$\sum_{\tau=0}^{\infty} |S_{ab}(\tau)|^2$  converges to a finite limit $\Omega_{ab}$, $\forall (a,b)\in\{1,\dots,p\}^2$, where $S_{ab}$ denote the entries of matrix $\matr{S}$. 
\end{assumption}


\subsection{Lemmas}
The estimated multivariate kurtosis (\ref{estimKurt-eq}) is a rational function of degree 4. Since we wish to calculate its asymptotic first and second order moments, when $N$ tends to infinity, we may expand this rational function about its mean. The first step is to expand the estimated covariance $\hat{\matr{S}}$. Let $\hat{\matr{S}}=\matr{S}+\matr{\Delta}$, where $\matr{\Delta}$ is small compared to $\matr{S}$; 
in fact : 
\begin{lemma}\label{lemma1}
The entries of matrix $\matr{\Delta}$ are of order $O(1/\sqrt{N})$. 
\end{lemma}
\begin{proof}\noindent Under Hypothesis $H$, the covariance of entries $\Delta_{ab}$ take the form below :
$$
Cov(\matr{\Delta}_{ab},\matr{\Delta}_{cd})= \frac{1}{N^{2}}\sum_{n=1}^N\sum_{m=1}^N \E{x_a(n) x_b(n) x_c(m) x_d(m)} - S_{ab} S_{cd} 
$$
and letting $\tau=n-m$, and $\Omega_{abcd}=S_{ac}S_{bd}+S_{ad}S_{bc}$ we have after some manipulation:
\begin{eqnarray*}
Cov(\matr{\Delta}_{ab},\matr{\Delta}_{cd}) &=& \frac{1}{N}\Omega_{abcd}
    +\frac{2}{N} \sum_{\tau=1}^{N-1}(1-\frac{\tau}{N})\left\{ S_{ac}(\tau)S_{bd}(\tau)+S_{ad}(\tau)S_{bc}(\tau) \right\}  \\
&\le& \frac{1}{N}\Omega_{abcd} + \frac{2}{N} \sum_\tau \left\{ |S_{ac}(\tau)|\,|S_{bd}(\tau)|+|S_{ad}(\tau)|\,|S_{bc}(\tau)| \right\}.
\end{eqnarray*}
Next, using the inequalities $|\sum_i u_i v_i|\le \sum_i |u_i| |v_i| \le \frac{1}{2}\sum_i (u_i^2 + v_i^2)$, we have:
$$
|Cov(\matr{\Delta}_{ab},\matr{\Delta}_{cd})| \le \frac{|\Omega_{abcd}|}{N} + \frac{1}{N} \sum_\tau |S_{ac}(\tau)|^2+|S_{bd}(\tau)|^2+|S_{ad}(\tau)|^2+|S_{bc}(\tau)|^2.
$$
Now using the mixing condition, $\sum_{\tau=0}^{\infty} |S_{ij}(\tau)|^2\le \Omega_{ij}$, we eventually obtain:
\begin{equation}
|Cov(\matr{\Delta}_{ab},\matr{\Delta}_{cd})| \le \frac{|\Omega_{abcd}|}{N} + \frac{1}{N} (\Omega_{ac}+\Omega_{bd} + \Omega_{ad}+\Omega_{bc})
\end{equation}
which shows that $Cov(\matr{\Delta}_{ab},\matr{\Delta}_{cd})=O(1/N)$.
\end{proof}


\begin{lemma}\label{lemma2}
The inverse $\hat{\matr{G}}$ of $\hat{\matr{S}}$ can be approximated by
\begin{equation}\label{eq:expansionG}
\hat{\matr{G}} = \matr{G} - \matr{G}\matr{\Delta}\matr{G} + 
\matr{G}\matr{\Delta}\matr{G}\matr{\Delta}\matr{G} + o(1/N). 
\end{equation}
\end{lemma}

\begin{proof}
  Notice that positive definite sample covariance matrix may be reexpressed as 
$$
\hat{\matr{S}} = \matr{S}+\matr{\Delta} = \matr{S}^{1/2}\matr{I}\matr{S}^{1/2}+\matr{S}^{1/2}\matr{S}^{-1/2}\matr{\Delta}\matr{S}^{-1/2}\matr{S}^{1/2}
$$
 
Let $\matr{E}$ be the symmetric  matrix $\matr{E} = -\matr{S}^{-1/2}\matr{\Delta}\matr{S}^{-1/2}$. Then with this definition, 
$$\hat{\matr{G}}=\matr{S}^{-1/2}(\matr{I}+\matr{E})^{-1}\matr{S}^{-1/2}$$ 
  As for   any matrix $\matr{E}$ with spectral radius smaller than 1, the series $\sum_{k=0}^\infty \matr{E}^k$ converges to $(\matr{I}-\matr{E})^{-1}$. 
If we plug this series in the expression of $\hat{\matr{G}}$, for $N$ large enough to warrant that the spectral radius of $\matr{E}$is less than 1,  we get 
$\hat{\matr{G}}=\matr{S}^{-1/2}\,\sum_{k=0}^K \matr{E}^k\,\matr{S}^{-1/2}+o(\|\matr{E}\|^K)$. Replacing $\matr{E}$  by its definition   and taking $K=3$   eventually yields (\ref{eq:expansionG}).   Note that the precise approximation order is $O(N^{-3/2})$, but only $o(1/N)$ will be useful in what follows.  
\end{proof}

Now it is desirable to express $\hat{\matr{G}}$ as a function of $\hat{\matr{S}}$. If we replace $\matr{\Delta}$ by $\hat{\matr{S}}-\matr{S}$ in (\ref{eq:expansionG}), we obtain:
\begin{equation}\label{eq:approxG}
\hat{\matr{G}}=3\matr{G} -3\matr{G}\hat{\matr{S}}\matr{G} +
\matr{G}\hat{\matr{S}}\matr{G}\hat{\matr{S}}\matr{G} + o(1/N). 
\end{equation}
With this approximation, $\hat{\matr{G}}$ is now a polynomial function of $\hat{\matr{S}}$ of degree 2, and hence of degree 4 in $\vect{x}$. We shall show that the mean of $\hat{B}_p(N)$ involves moments of $\vect{x}$ up to order 8, whereas its variance involves moments up to order 16.

\begin{lemma}\label{lem:DLBp}
Denote $A_{ij}=\vect{x}(i)^T\matr{S}^{-1}\vect{x}(j)$. Then:
\begin{equation}\label{eq:DLBp}
\begin{split}
     \hat{B}_p(N)& = \frac{6}{N}\sum_{n=1}^{N} A_{nn}^2 -\frac{8}{N^2}\sum_{n=1}^{N}A_{nn}\sum_{i=1}^{N}A_{ni}^2 + \frac{1}{N^3} \sum_{n=1}(\sum_{i=1}^{N}A_{ni}^2)(\sum_{j=1}^{N}A_{nj}^2)\\
     &
     +\frac{2}{N^3}\sum_{n=1}^{N}\sum_{j=1}^{N}\sum_{k=1}^{N} A_{nn}A_{nj}A_{jk}A_{kn} + o(1/N)
     \end{split}
\end{equation}
\end{lemma}

\begin{proof}
First inject (\ref{eq:expansionG}) in the expression  $\hat{B}_p(N)=\frac{1}{N}\sum_{n=1}^N \left(\vect{x}(n)^T\hat{\matr{G}}\vect{x}(n)\right)^2$, and  keep terms up to order $O(\|\matr{\Delta}\|^2)$; this yields:
\begin{eqnarray*}
\hat{B}_p(N)=\frac{1}{N} \sum_n \left[A_{nn}^2   
 -2A_{nn}\,\vect{x}(n)^T\matr{G}\matr{\Delta}\matr{G}\vect{x}(n)
 +\big(\vect{x}(n)^T\matr{G}\matr{\Delta}\matr{G}\vect{x}(n)\big)^2 \right.\\
 \left.  + 2A_{nn}\,\vect{x}(n)^T\matr{G}\matr{\Delta}\matr{G}\matr{\Delta}\matr{G}\vect{x}(n)\right]
 +o(\|\matr{\Delta}\|^2).
\end{eqnarray*}
Then replace $\matr{\Delta}$ by $\hat{\matr{S}}-\matr{S}$. This leads to 
\begin{eqnarray*}
\hat{B}_p(N)=\frac{1}{N} \sum_n \left[
6A_{nn}^2 -8 A_{nn}\big(\vect{x}(n)^T\matr{G}\hat{\matr{S}}\matr{G}\vect{x}(n)\big) \right.\\
\left. +\big(\vect{x}(n)^T\matr{G}\hat{\matr{S}}\matr{G}\vect{x}(n)\big)^2 
+ 2 A_{nn} \big(\vect{x}(n)^T\matr{G}\hat{\matr{S}}\matr{G}\hat{\matr{S}}\matr{G}\vect{x}(n)\big)
\right] +o(\|\matr{\Delta}\|^2).
\end{eqnarray*}
Equation (\ref{eq:DLBp}) is eventually obtained after replacing $\hat{\matr{S}}$ by $\frac{1}{N}\sum_k\vect{x}(k)\vect{x}(k)^T$ and all terms of the form $\vect{x}(q)^T\matr{G}\vect{x}(r)$ by $A_{qr}$. 
\end{proof}

\subsection{Additional notations and   calculus   issues}

When computing the mean and variance of $\hat{B}_p(N)$ given in (\ref{eq:DLBp}), higher order moments of the multivariate random variable $\vect{x}$ will arise. Under the normal (null) hypothesis, these moments are expressed as functions of second order moments only. To keep notations reasonably concise, it is proposed to use McCullagh's bracket notation \citep{Mccu87}, briefly reminded in Appendix \ref{Appendix-A1}. Furthermore, for all moments of order higher than $p$, some components appear multiple times; counting the number of identical terms in the expansion of the higher moments is a tedious task. All the moment expansions  that are necessary for the derivations presented in this paper are developed in Appendix \ref{app:moments}.

In order to keep notations as explicit and concise as possible, while keeping explicit the role of both coordinate (or space) indices and time indices,  let 
the  moments of $\vect{x}(t)$, whose $p$ components are $x_a(t)$, $1\le a\le p$ be noted
\begin{equation}\label{eq:mu}
\mu_{ab}^{tu}=\E{x_a(t)x_b(u)}, \quad
\mu_{abc}^{tuv}=\E{x_a(t)x_b(u)x_c(v)}
\end{equation}
and so forth for higher orders. It shall be emphasized that different time and coordinate indices appear here as the components are assumed to be colored (time correlated) and dependent  to each others (spatially correlated). 

Computation of the mean and variance  of $\hat B_p$ defined by equation (\ref{eq:DLBp}) involves the computation of moments of order  noted $2L$  whose generic expression is 

\begin{equation*}
\E{\prod_{l=1}^L A_{\alpha_l\beta_l}}=\sum_{r_1\ldots r_L,c_1\ldots c_L =1}^p \left(\prod_{i=1}^{L}G_{r_i,c_i} \right) \E{ x_{r_1}({\alpha_1})x_{c_1}({\beta_1}) \ldots x_{r_L}(\alpha_L)x_{c_L}(\beta_L)}
\end{equation*}

or equivalently

\begin{equation}\label{eq:Annex_1}
\E{\prod_{l=1}^L A_{\alpha^l\beta^l}}=\sum_{r_1\ldots r_L,c_1\ldots c_L =1}^p \left(\prod_{i=1}^{L}G_{r_i,c_i} \right) 
\mu_{r_1\ldots r_L c_1\ldots c_L }^{\alpha_1\ldots\alpha_L\beta_1\ldots\beta_L}
\end{equation}

In the above equation,  the $2L$-order moment $\mu_{r_1\ldots r_L c_1\ldots c_L }^{\alpha_1\ldots\alpha_L\beta_1\ldots\beta_L}$ has superscripts indicating the time indices involved, whereas the subscripts indicate the coordinate (or space) indices. 

While being general, the above formulation may take simpler, or more explicit forms in practice. 
The  detailed methodology for computing the expressions of the mean and variance of  $\hat B_p$ as functions of second order moments is deferred to Appendix \ref{Appendix-A2}. The resulting expressions of Mardia's statistics  are given and discussed in the sections to come. 

\section{Expression of the mean of   $\widehat{B}_p(N)$   }\label{sec-genexp-mean}
According to Equation (\ref{eq:DLBp}), we have four types of terms. The goal of this section is to provide the expectation of each of these terms.   In the propositions below, all terms are developed as being sums and products of second order moments, as it is reminded that under $H$ the process is Gaussian. 
Notice also that under the latter assumption, all higher-order moments of any order are finite. For sake of simplicity,  Landau's approximation order  $O(h(n))$ is omitted in most equations.  

\begin{lemma}\label{lem:meanAA}
With the definition of $A_{ij}$ given in Lemma \ref{lem:DLBp}, we have:
\begin{eqnarray}
\E{A_{nn}^2} &=& \sum_{a,b,c,d=1}^{p} G_{ab}G_{cd} \,\mu^{nnnn}_{abcd}
\\ 
\E{A_{nn} A_{ni}^2} &=& \sum_{a,b,c,d=1}^{p}\,\sum_{e,f=1}^{p}G_{ab}G_{cd}G_{ef} \,\mu_{abcedf}^{nnnnii}
\\
\E{A_{ni}^2A_{nj}^2} &=& \sum_{a,b,c,d=1}^{p}\,\sum_{e,f,g,h=1}^{p} G_{ab}G_{cd}G_{ef}G_{gh}\,\mu_{acegbdfh}^{nnnniijj} 
\\
\E{A_{nn}A_{nj}A_{jk}A_{kn}}&=&\sum_{a,b,c,d=1}^{p}\,\sum_{e,f,g,h=1}^{p} G_{ab}G_{cd}G_{ef}G_{gh}\,\mu_{abchdefg}^{nnnnjjkk}    
\end{eqnarray}
\end{lemma}

\begin{proposition}\label{mean-prop}
Using expressions of moments given in Appendix \ref{app:moments}, the expectations of the four terms defined in Lemma \ref{lem:meanAA} take the form below
\begin{eqnarray}
\E{A_{nn}^2} &=& \displaystyle\sum_{k\ell qr=1}^{p} G_{k\ell}G_{rq} \Big\{ [3]\mu^{nn}_{k\ell}\mu^{nn}_{qr}\Big\}\\
\E{A_{nn} A_{ni}^2} &=& \displaystyle\sum_{k\ell qrst=1}^{p} G_{k\ell}G_{qr}G_{st} \Big\{[12] \mu^{ni}_{kr} \mu^{ni}_{\ell t} \mu^{nn}_{qs} 
+ [3] \mu^{nn}_{k\ell} \mu^{nn}_{qs} \mu^{ii}_{rt} \Big\}
\\
\E{A_{ni}^2 A_{nj}^2} &=& \displaystyle \sum_{k,\ell,q,r}\,\sum_{s,t,u,v} G_{k\ell} G_{qr} G_{st} G_{uv} \Big\{[3] \mu^{nn}_{kq} \mu^{nn}_{su} \mu^{ii}_{\ell r} \mu^{jj}_{tv} \nonumber\\
&&+[6] \mu^{nn}_{kq} \mu^{nn}_{su} \mu^{ij}_{\ell t}\mu^{ij}_{rv} 
+[12]\mu^{nn}_{kq} \mu^{ni}_{s\ell} \mu^{ni}_{ur} \mu^{jj}_{tv} 
+[24] \mu^{nj}_{kt} \mu^{nj}_{qv} \mu^{in}_{\ell s} \mu^{ni}_{ur} \nonumber\\
&&+[48] \mu^{ni}_{k\ell}\mu^{ij}_{rt} \mu^{nj}_{qv} \mu^{nn}_{su}
+[12] \mu^{nn}_{kq} \mu^{jn}_{ts} \mu^{nj}_{uv} \mu^{ii}_{r\ell} \Big\} 
\\
\E{A_{nn}A_{nj}A_{jk}A_{kn}} &=& \displaystyle \sum_{m,\ell,q,r}\,\sum_{s,t,u,v} G_{m\ell} G_{qr} G_{st} G_{uv} \Big\{ [3] \mu^{nn}_{m\ell}\mu^{nn}_{qv}\mu^{jj}_{sr}\mu^{kk}_{tu} \nonumber\\
&&+ [6] \mu^{nn}_{m\ell}\mu^{nn}_{qv}\mu^{jk}_{rt}\mu^{jk}_{su} 
+ [12] \mu^{nn}_{m\ell}\mu^{nj}_{qr}\mu^{nj}_{vs}\mu^{kk}_{tu} 
+ [24] \mu^{nk}_{mv} \mu^{nk}_{\ell u} \mu^{nj}_{qr} \mu^{nj}_{vs} \nonumber\\
&&+[48] \mu^{nj}_{mr} \mu^{jk}_{st} \mu^{nk}_{\ell u} \mu^{nn}_{qv} 
+ [12]\mu^{nn}_{k\ell}\mu^{nk}_{qt}\mu^{nk}_{vu}\mu^{jj}_{rs} \Big\} 
\end{eqnarray}
\end{proposition}

 The mean of $\hat{B}_p(N)$ then follows from (\ref{eq:DLBp}). 


\section{Expression of the variance of   $\widehat{B}_p(N)$  }\label{sec-genexp-variance}

From Lemma \ref{lem:DLBp}, we can also state what moments of $A_{ij}$ will be required in the expression of the variance of $B_p(N)$. 
\begin{lemma}\label{lem:varAA}
 By  raising (\ref{eq:DLBp}) to the second power and using the definition of $A_{ij}$ given in Lemma \ref{lem:DLBp}, we  can check that the following moments are required:
\begin{eqnarray}
\E{A_{nn}^2 A_{ii}^2} &=& \sum_{a,b,c,d,e,f,g,h=1}^{p} G_{ab}G_{cd}G_{ef}G_{gh} \,\mu^{nnnniiii}_{abcdefgh}
\\ 
\E{A_{nn}^{2}A_{ij}^{2}A_{ii}} &=& \sum_{a,b,c,d=1}^{p}\,\sum_{e,f,g,h=1}^{p}\,\sum_{m,\ell=1}^{p} G_{ab}G_{cd}G_{ef}G_{gh}\nonumber \\&& G_{m\ell}\,\mu_{abcdegm\ell fh}^{nnnniiiijj}
\\
\E{A_{nn}A_{kk}A_{ni}^{2}A_{kj}^{2}} &=& \sum_{a,b,c,d=1}^{p}\,\sum_{e,f,g,h=1}^{p}\,\sum_{m,\ell,q,r=1}^{p} G_{ab}G_{cd}G_{ef}G_{gh} G_{m\ell} \nonumber \\&& G_{qr}\,\mu_{abegcdmqfh\ell r}^{nnnnkkkkiijj} \\
\E{A_{kk}^{2}A_{ni}^{2}A_{nj}^{2}}&=& \sum_{a,b,c,d=1}^{p}\,\sum_{e,f,g,h=1}^{p}\,\sum_{m,\ell,q,r=1}^{p} G_{ab}G_{cd}G_{ef}G_{gh} G_{m\ell} \nonumber \\&& G_{qr}\,\mu_{abcdegmqfh\ell r}^{kkkknnnniijj} \\
\E{A_{nn}^{2}A_{ii}A_{ij}A_{jk}A_{ki}}&=&\sum_{a,b,c,d=1}^{p}\,\sum_{e,f,g,h=1}^{p}\,\sum_{m,\ell,q,r=1}^{p} G_{ab}G_{cd}G_{ef}G_{gh} G_{m\ell} \nonumber \\&& G_{qr}\,\mu_{abcdefgrhm\ell q}^{nnnniiiijjkk}  
\\
\E{A_{ni}^{2}A_{nj}^{2}A_{kt}^{2}A_{kk}} &=&\sum_{a,b,c,d=1}^{p}\,\sum_{e,f,g,h=1}^{p}\,\sum_{m,\ell,q,r=1}^{p}\,\sum_{s,u=1}^{p}\,G_{ab}G_{cd}G_{ef}G_{gh}G_{m\ell} \nonumber \\&& G_{qr}G_{su}\,\mu_{acegmqsubdfh\ell r}^{nnnnkkkkiijjtt}
\end{eqnarray}
\begin{eqnarray}
\E{A_{ii}A_{it}^{2}A_{nn}A_{nj}A_{jk}A_{kn}} &=&\sum_{a,b,c,d=1}^{p}\,\sum_{e,f,g,h=1}^{p}\,\sum_{m,\ell,q,r=1}^{p}\,\sum_{s,u=1}^{p}G_{ab}G_{cd}G_{ef}G_{gh}G_{m\ell} \nonumber \\&& G_{qr}G_{su}\,\mu_{abceghmu\ell qrsdf}^{iiiinnnnjjkktt}
\\
\E{A_{ni}^{2}A_{kt}^{2}A_{nj}^{2}A_{ku}^{2}} &=&\sum_{a,b,c,d=1}^{p}\,\sum_{e,f,g,h=1}^{p}\,\sum_{m,\ell,q,r=1}^{p}\,\sum_{s,v,w,z=1}^{p}G_{ab}G_{cd}G_{ef}G_{gh}G_{m\ell} \nonumber \\&& G_{qr}G_{sv}G_{wz}\,\mu_{acmqegswbd\ell rfhvz}^{nnnnkkkkiijjttuu}
\\
\E{A_{nn}A_{nj}A_{jk}A_{kn}A_{ii}A_{it}A_{tu}A_{ui}} &=&\sum_{a,b,c,d=1}^{p}\,\sum_{e,f,g,h=1}^{p}\,\sum_{m,\ell,q,r=1}^{p}\,\sum_{s,v,w,z=1}^{p}G_{ab}G_{cd}G_{ef}G_{gh}G_{m\ell} \nonumber \\&& G_{qr}G_{sv}G_{wz}\,\mu_{abchm\ell qzdefgrsvw}^{nnnniiiijjkkttuu}
\\
\E{A_{nn}A_{nj}A_{jk}A_{kn}A_{it}^{2}A_{iu}^{2}} &=&\sum_{a,b,c,d=1}^{p}\,\sum_{e,f,g,h=1}^{p}\,\sum_{m,\ell,q,r=1}^{p}\,\sum_{s,v,w,z=1}^{p}G_{ab}G_{cd}G_{ef}G_{gh}G_{m\ell} \nonumber \\&& G_{qr}G_{sv}G_{wz}\,\mu_{abchmqswdefg\ell rvz}^{nnnniiiijjkkttuu}
\end{eqnarray}
\end{lemma}

Then, as in Proposition \ref{mean-prop}, by using the results of Appendix \ref{app:moments}, the moments $\mu^*_*$ could be in turn expressed as a function of second order moments. For readability, we do not substitute here these values.

\begin{proposition}
\begin{eqnarray}
     Var\{\hat{B}_{p}\}&=&\displaystyle \frac{36}{N^2}\sum_{n}\sum_{i}\E{A_{nn}^2 A_{ii}^2}-\frac{96}{N^3}\sum_{j}\sum_{n,i}\E{A_{nn}^2 A_{ij}^2A_{ii}}\nonumber \\&& +\frac{64}{N^4} \sum_{n,i}\sum_{j,k}\E{A_{nn}A_{kk}A_{ni}^2 A_{kj}^2} + \frac{12}{N^4} \sum_{n,i,j,k}\E{A_{kk}^{2}A_{ni}^{2}A_{nj}^{2}}\nonumber\\&& + \frac{24}{N^4} \sum_{n,i,j,k} \E{A_{nn}^{2} A_{ii}A_{ij}A_{jk}A_{ki}} - \frac{16}{N^5} \sum_{n,i} \sum_{j,k,t} \E{A_{ni}^{2}A_{nj}^{2}A_{kt}^{2}A_{kk}} \nonumber\\&& - \frac{32}{N^5}\sum_{i,t}\sum_{n,j,k}\E{A_{ii}A_{it}^{2}A_{nn}A_{nj}A_{jk}A_{kn}}+ \frac{1}{N^6} \sum_{n,i,j}\sum_{k,t,u}\E{ A_{ni}^{2}A_{kt}^{2}A_{nj}^{2}A_{ku}^{2}} \nonumber\\&&  +\frac{4}{N^6} \sum_{n,j,k}\sum_{i,t,u} \E{A_{nn}A_{nj}A_{jk}A_{kn}A_{ii}A_{it}A_{tu}A_{ui}}\nonumber\\&&+\frac{4}{N^6} \sum_{n,j,k}\sum_{i,t,u} \E{A_{nn}A_{nj}A_{jk}A_{kn}A_{it}^{2}A_{iu}^{2}}- (\E{\hat{B}_{p}})^{2} 
\end{eqnarray}
\end{proposition} 

\section{Mean and variance of $\hat B_1(N)$ in the scalar case $(p=1)$}\label{sec-scalexps}

The complicated expressions obtained in the previous sections simplify drastically in the scalar case, and we get the results below.
\begin{equation}\label{meanb1}
    \boxed{\E{\hat{B}_{1}} = 3 - \frac{6}{N} - \frac{12}{N^2} \sum_{\tau=1}^{N-1} (N-\tau) \,\frac{S(\tau)^2}{S^2}  + o(\frac{1}{N})  }
\end{equation}
\begin{equation}\label{varb1}
    \boxed{Var\{\hat{B}_{1}\} = \frac{24}{N}\Big[1+\frac{2}{N}\sum_{\tau=1}^{N-1}(N-\tau)  \frac{S(\tau)^{4}}{S^{4}}\Big]   + o(\frac{1}{N})  }
\end{equation} 
 In particular in the i.i.d. case, $S(\tau)=0$ for $\tau\neq0$, and we get the well-known result \citep{Mard74:sank} \citep{ComoD95:begur}:
$$
\E{\hat{B}_{1}} \approx 3 - \frac{6}{N}, \quad \text{and} \quad Var\{\hat{B}_{1}\} \approx \frac{24}{N}.
$$
  The expressions of mean and variance above are identical to  those given in Theorem \ref{mk}, the difference being that here the ratio $\frac{N-1}{N+1}$ is replaced by its approximation of order $N^{-1}$, i.e. $\frac{N-1}{N+1} = 1-\frac{2}{N}+ o(1/N)$.

\section{Mean and variance of $\widehat B_2(N)$ in the bivariate case $(p=2)$} \label{sec-bivexps}
 In the bivariate case, expressions become immediately more complicated, but we can still write them explicitly, as reported below. We remind that $\mu_{ab}^{ij}=S_{ab}(i-j)$.
\begin{equation}
    \boxed{\E{\hat{B}_{2}} = 8 - \frac{16}{N} -  \frac{4}{N^2}\sum_{\tau=1}^{N-1} \frac{(N-\tau) Q_{1}(\tau)}{(S_{11}S_{22}-S_{12}^2)^2} + o(\frac{1}{N}) }
\end{equation}
 with 
\begin{eqnarray}
Q_{1}(\tau) &=& S_{11}S_{22}\Big[\, (S_{12}(\tau)+S_{21}(\tau))^2 - 4S_{11}(\tau)S_{22}(\tau)\,\Big] \nonumber \\ && + S_{12}^{2}\Big[\,2(S_{12}(\tau)+S_{21}(\tau))^{2}+4S_{22}(\tau)S_{11}(\tau)\,\Big]\nonumber \\ &&  -6 S_{22}S_{12}\Big(S_{11}(\tau)(S_{12}(\tau)+S_{21}(\tau))\Big) \nonumber \\ &&  -6 S_{11}S_{12}\Big(S_{22}(\tau)(S_{12}(\tau)+S_{21}(\tau)\Big) \nonumber \\ &&
+6S_{11}^{2}S_{22}^{2}(\tau) +6S_{22}^{2}S_{11}^{2}(\tau). 
\end{eqnarray}
\begin{equation}
    \boxed{
    Var\{\hat{B}_{2}\}=\frac{64}{N}+\frac{16}{N^{2}}\sum_{\tau=1}^{N-1}\frac{(N-\tau)Q_{2}(\tau)}{(S_{11}S_{22}-S_{12}^{2})^4} + o(\frac{1}{N})  } 
\end{equation}
with 
\begin{eqnarray}\label{draft}
Q_{2}(\tau)&=&\Big[2S_{11}^{2}(\tau)S_{22}^{2}(\tau)-16S_{11}(\tau)S_{22}(\tau)S_{12}(\tau)S_{21}(\tau)\nonumber\\&&+3(S_{21}^{2}(\tau)+S_{12}^{2}(\tau))^2+12S_{11}(\tau)S_{22}(\tau)(S_{12}(\tau)+S_{21}(\tau))^2\nonumber \\ &&-4S_{12}^{2}(\tau)S_{21}^{2}(\tau) \Big]S_{11}^{2}S_{22}^{2} \nonumber\\&&+2 S_{11}^2 S_{12}^{2}\Big[8S_{11}(\tau)S_{22}(\tau)+3(5S_{11}(\tau)S_{22}(\tau)+S_{21}(\tau)S_{12}(\tau))\nonumber\\ && (S_{21}(\tau)+S_{12}(\tau))^{2} - 4S_{21}(\tau)S_{12}(\tau)\big(S_{22}^{2}(\tau)+S_{21}(\tau)S_{12}(\tau)\big)\Big] \nonumber \\ && + 2 S_{22}^2 S_{12}^{2}\Big[8S_{22}(\tau)S_{11}(\tau)+3(5S_{11}(\tau)S_{22}(\tau))+S_{21}(\tau)S_{12}(\tau)\nonumber\\ && (S_{21}(\tau)+S_{12}(\tau))^{2} - 4S_{21}(\tau)S_{12}(\tau)\big(S_{11}^{2}(\tau)+S_{21}(\tau)S_{12}(\tau)\big) \Big]\nonumber \\ &&
+ 3S_{11}^{4}S_{22}^{4}(\tau) + 3S_{22}^{4}S_{11}^{4}(\tau) \nonumber\\ && + 8S_{12}^{4}\Big[S_{11}^{2}(\tau)S_{22}^{2}(\tau)+4S_{11}(\tau)S_{22}(\tau)S_{12}(\tau)S_{12}(\tau)+S_{21}^{2}(\tau)S_{12}^{2}(\tau)\Big]\nonumber\\&& -12S_{11}S_{12}\, S_{22}(\tau)(S_{12}(\tau)+S_{21}(\tau)\Big[(2S_{11}(\tau)S_{22}(\tau)+S_{21}^{2}(\tau)+S_{12}^{2}(\tau))\nonumber \\&&S_{11}S_{22}+2(S_{11}(\tau)S_{22}(\tau)+S_{12}(\tau)S_{21}(\tau))S_{12}^{2}\Big] \nonumber\\&& -12S_{22}S_{12}\, S_{11}(\tau)(S_{12}(\tau)+S_{21}(\tau)\Big[(2S_{11}(\tau)S_{22}(\tau)+S_{21}^{2}(\tau)+S_{12}^{2}(\tau))\nonumber \\&&S_{11}S_{22}+2\big(S_{11}(\tau)S_{22}(\tau)+S_{12}(\tau)S_{21}(\tau)\big)S_{12}^{2}\Big]. 
\end{eqnarray}
Note that the latter expressions are complicated, but easy to implement as demonstrated in the remaining sections.   Again for this case where $p=2$, the approximation $\frac{N-1}{N+1} = 1-\frac{2}{N}+ o(1/N)$ was used in the expressions of the mean and variance of $\hat B_2$.  

\section{Particular case: multidimensional embedding of a scalar process}\label{sec-embed-exps}
 In this section, we consider the particular case where the multivariate process consists of the embedding of a scalar process. More precisely, we assume that

$$
\vect{x}(n) = \left(\begin{array}{c} x_1(n)\\ \dots\\ x_p(n)\end{array} \right)
= \left(\begin{array}{c} y(n\delta+1)\\ \dots\\ y(n\delta+p)\end{array} \right).
$$
where $y(k)$ is a scalar wide-sense stationary process of correlation function $C(\tau)=\E{y(k)y(k-\tau)}=S_{11}(\tau/\delta)$.
Note that now, because of the particular form of $\vect{x}(n)$, we can exploit the translation invariance by  remarking that $S_{ab}(\tau)=\E{x_a(n\delta)x_b(n\delta-\tau\delta)}$ implies $S_{ab}(\tau)=C(\tau\delta+a-b)$, for $1\le a,b\le p$.

To keep results as concise as possible, we  assume the notation
$\gamma_i(\tau)=C(\tau\delta+i)$, and the shortcut $C_j=C(j)$. The main goal targeted by defining these multiple notations is to obtain more compact expressions.

\subsection{Bivariate embedding}
The bivariate case is more difficult but the expressions still have a simple form:
\begin{equation}\label{meanb2}
    \boxed{\E{\hat{B}_{2}} \approx 8 - \frac{16}{N} -  \frac{4}{N^2}\sum_{\tau=1}^{N-1} \frac{(N-\tau) q_1(\tau)}{(C_{0}^2-C_{1}^{2})^2}}
\end{equation}
\begin{equation}\label{varb2}
    \boxed{
    Var\{\hat{B}_{2}\}\approx \frac{64}{N}+\frac{16}{N^{2}}\sum_{\tau=1}^{N-1}\frac{(N-h)q_2(\tau)}{(C_{0}^{2}-C_{1}^{2})^4}} 
\end{equation}

\noindent with $q_1(\tau)$ and $q_2(\tau)$  defined below, where $\gamma_i$ stands for  $\gamma_{i}(\tau)$:\noindent
\begin{equation}
    \begin{split}
    q_1(\tau) &= \Big[(\gamma_{1} + \gamma_{-1})^2 + 8 \gamma_{0}^2\Big ]C_{0}^2
    -12 C_{0}C_{1}\, \gamma_{0} (\gamma_{1} + \gamma_{-1}) \\&+ \Big[2 (\gamma_{1} + \gamma_{-1})^2 + 4 \gamma_{0}^2\Big]C_{1}^2,  
\end{split}
\end{equation}
\begin{equation}
\begin{split}
q_2(\tau)&=\Big[8(\gamma_{0}^{2}-\gamma_{1}\gamma_{-1})^{2}+3(\gamma_{1}^{2}-\gamma_{-1}^{2})^{2}+12 \gamma_{0}^{2}(\gamma_{1}+\gamma_{-1})^{2}\Big]C_{0}^{4} \\&+ 4\,\Big[8\gamma_{0}^{4}+3\,(5\gamma_{0}^{2}+\gamma_{1}\gamma_{-1})(\gamma_{1}+\gamma_{-1})^{2}-4\gamma_{1}\gamma_{-1}(\gamma_{0}^{2}+\gamma_{1}\gamma_{-1})\Big]C_{0}^{2}C_{1}^{2} +\\& 8\,\Big[\gamma_{0}^{4}+4\gamma_{0}^{2}\gamma_{1}\gamma_{-1}+\gamma_{1}^{2}\gamma_{-1}^{2}\Big]C_{1}^{4}-24C_{0}C_{1}\, \gamma_{0}(\gamma_{1}+\gamma_{-1})\Big[(2\gamma_{0}^{2}+\gamma_{1}^{2}+\gamma_{-1}^{2})C_{0}^{2}\\+&2(\gamma_{0}^{2}+\gamma_{1}\gamma_{-1})C_{1}^{2}\Big].
\end{split}
\end{equation}
The exact computation for the trivariate embedding case have also been conducted; but because of their lengthy expressions (especially that of the variance), they are not detailed here and can be given as supplementary material upon request.

\section{Computer experiments}\label{sec-comp-exp}
In this section, the preceding results are illustrated on dedicated computer experiments. To emphasize the importance of the univariate and the bivariate normality tests on colored random process, we simulate \emph{correlated} bivariate random processes with Gaussian marginals. The generation procedure is briefly described in the next section. Then tests are performed to detect non Gaussian nature of the joint distribution while the marginals remain Gaussian.

\begin{remark}
Up to now, we have derived the mean and variance of a test variable $\widehat{B}_p$. In order to compute the power of the test, we need its distribution. First, $\widehat{B}_p$ is shown in \citep{Mard70:biom} to converge to ${B}_p$ in probability. Next, ${B}_p$ is a sum of n.i.d. random variables enjoying the \emph{mixing} Property \ref{mixing-ass}; for this reason ${B}_p$ converges to a normal variable thanks to the Law of Large Numbers \citep[ch.IV]{Hann70}. This guarantees that $(\widehat{B}_p-\E{\widehat{B}_p})/\sqrt{VAR\{\widehat{B}_{p}\}}$ is asymptotically $\mathcal{N}(0,1)$.
\end{remark}

\subsection{Gaussian Marginals under $\bar{H}$}
Copulas are a   classical framework, which is simple to implement  
 for defining multivariate distributions with controlled joint distribution function. It is known that there is a \textit{unique} copula -- called the Gaussian copula $\mathcal{C}_{R}$ -- that  produces the bivariate Gaussian distribution, fully specified by the correlation matrix $R$:\\ 

\begin{equation}
\mathcal{C}_{R}(u,v) = \int_{-\infty}^{\Phi^{-1}(u)}\int_{-\infty}^{\Phi^{-1}(v)} \frac{1}{2\pi (1-R_{12}^{2})^{1/2}} exp\Big\{ \frac{s^{2}-2R_{12}st+t^{2}}{2(1-R_{12}^{2})}\Big\}\,ds\,dt 
\end{equation}
where $\Phi^{-1}$ 
is the inverse of the cumulative distribution function of the standard normal distribution.
As Sklar's theorem (cf. Appendix \ref{Sklar}) guarantees the  uniqueness of the copula generating a given bivariate distribution, non Gaussian distributions can easily be obtained by using other types of copulas. Namely here, Clayton and Gumbel bivariate copulas are used as examples: 



\begin{eqnarray}
\textit{Clayton: }\mathcal{C}_{\theta}(u,v)&=& max\{u^{-\theta}+v^{-\theta}-1;0\}, \theta \in [-1,\infty)\char`\\ \{0\} \\
\textit{Gumbel: }\mathcal{C}_{\theta}(u,v)&=& exp\Big\{ -(-log(u)^{\theta} + -log(v)^{\theta} )^\frac{1}{\theta}\Big\}, \theta \in [1,\infty)
\end{eqnarray}
\FIN

Since Sklar's theorem  does not impose \textit{independence} of any variate $u$ or $v$ of $\mathcal{C}_{\theta}(u,v)$, we   need to   propose the following algorithm to generate a bivariate copula with \textit{colored} Gaussian marginals. 
\begin{itemize}
    \item Generate two \textit{i.i.d} centered normalized Gaussian variables: $\eta_{1},\eta_{2} \underset{i.i.d}{\sim} \mathcal{N}(0,1)$ 
    \item Make the previous variables correlated in time by a first-order auto-regressive filter:
    \begin{eqnarray}
    y_{1}(n)&=&0.8 y_{1}(n-1)+\eta_{1}(n) \nonumber \\
    y_{2}(n)&=&0.8 y_{2}(n-1)+\eta_{2}(n) \nonumber
    \end{eqnarray}
    Thus $\E{y_{1}(n)y_{1}(n-k)} = 0.8^{|k|}$, for all $k \in \mathbb{Z}$.
    \item Transform $y_{1}$ and $y_{2}$ as:
    \begin{eqnarray}
        u &=& \Phi(y_{1}) \\
        v &=& \Phi(y_{2}) 
    \end{eqnarray}
    Note that $u$ and $v$ are uniformly distributed on $[0,1]$. Thus, we can generate new samples $u^{\prime}$, $v^{\prime}$ coupled by a given copula $\mathcal{C}_{\theta}$. For more details about efficient sampling of copula see the (Marshall and Olkin 1988 algorithm) cited in \citep{Hofe08:csda}.
    
    \item Transform $u^{\prime}$ and $v^{\prime}$ to obtain Gaussian standard marginals:  $\vect{x}=(x_1(n),\, x_{2}(n))^{T}:$
    \begin{eqnarray}
    x_{1}(n)&=&\Phi^{-1}\big(u(n)\big) \nonumber \\
    x_{2}(n)&=&\Phi^{-1}\big(v(n)\big) \nonumber
    \end{eqnarray}
\end{itemize}

\subsection*{Simulation study} 
For a given copula $\mathcal{C}$, we perform $M=2000$ realizations of $\vect{x}(n)=(x_{1}(n), x_{2}(n))^{T}$ of total length $N=1000$.
First, the $p$-values of the two-sided tests are computed based on:
\begin{equation*}
    t = \frac{\hat{B}_{(.)}-\E{\hat{B}_{(.)}}}{\sqrt{\text{Var}\{\hat{B}_{(.)}\}}}
\end{equation*}
Recall that this statistic is standard normal. Then $p$-value = $2(1-\Phi(|t|))$ is compared to pre-specified significance levels $\alpha$.
For any $p$ smaller than $\alpha$, it is considered heuristically that the test rejected $H$.
The empirical rejection rates, defined by $\frac{\text{Number of rejections}}{M}$ for each statistic $\hat{B}_{1,i.i.d}, \hat{B}_{1}$ and $\hat{B_{2}}$ are reported in Table \ref{table:CopulasPerf}.\\

\begin{table}[htbp]
\begin{tabular}{l|llllll}
\multirow{2}{*}{\begin{tabular}[c]{@{}l@{}}Test \\ statistic\end{tabular}}                  & \multicolumn{2}{l|}{Gaussian $R_{12} = 0.8$}                                                                                                        & \multicolumn{2}{l|}{Clayton $\theta=2$}                                                                                             & \multicolumn{2}{l}{Gumbel $\theta=5$}                                                                                              \\ \cline{2-7} 
                                                                                            & $\alpha=5\%$                                                      & \multicolumn{1}{l|}{$\alpha=10\%$}                               & $\alpha=5\%$                                                     & \multicolumn{1}{l|}{$\alpha=10\%$}                               & $\alpha=5\%$                                                     & \multicolumn{1}{l}{$\alpha=10\%$}                               \\ \hline
\begin{tabular}[c]{@{}l@{}}$\hat{B}_{1,i.i.d}$\\ $\hat{B}_{1}$\\ $\hat{B}_{2}$\end{tabular} & \begin{tabular}[c]{@{}l@{}}0.1660 \\ 0.0450\\ 0.0480\end{tabular} & \begin{tabular}[c]{@{}l@{}}0.2460\\ 0.0730\\ 0.0801\end{tabular} & \begin{tabular}[c]{@{}l@{}}0.1011\\ 0.1060\\ 0.9890\end{tabular} & \begin{tabular}[c]{@{}l@{}}0.1651\\ 0.1701\\ 0.9920\end{tabular} & \begin{tabular}[c]{@{}l@{}}0.1189 \\ 0.0390\\ 0.9920\end{tabular} & \begin{tabular}[c]{@{}l@{}}0.1930\\ 0.0860\\ 0.9960\end{tabular} \\ \hline
\end{tabular}
\caption{Empirical Rejection rate at two significance levels : $\alpha=5\%,10\%$}
\label{table:CopulasPerf}
\end{table}
\begin{figure}[htbp]
\centering
\subfloat[Gaussian Copula $R_{12}=.8$]{\label{fig:d}\includegraphics[width=0.45\textwidth]{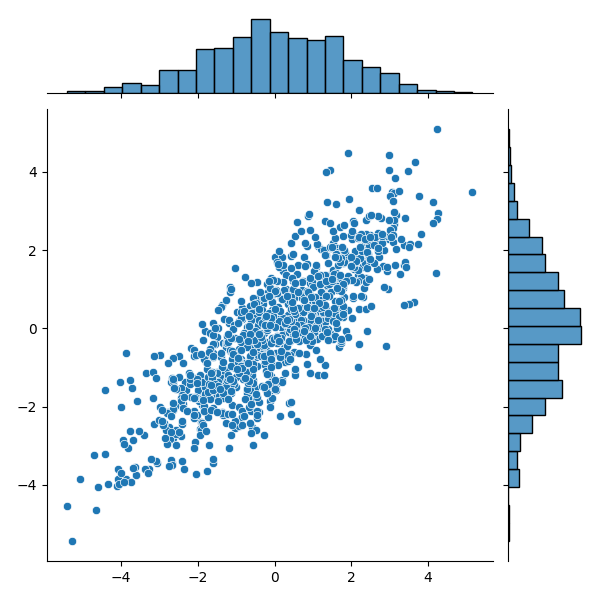}}\qquad%
\subfloat[Clayton $\theta=2$]{\label{fig:b}\includegraphics[width=0.45\linewidth]{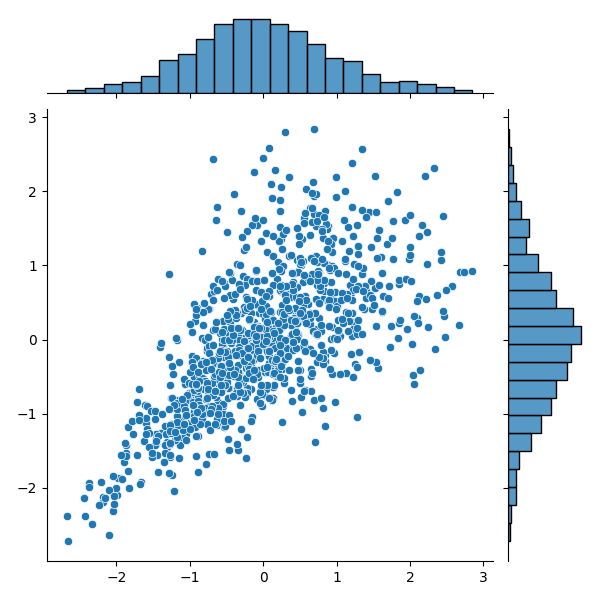}}\qquad%
\subfloat[Gumbel $\theta=5$]{\label{fig:c}\includegraphics[width=0.45\textwidth]{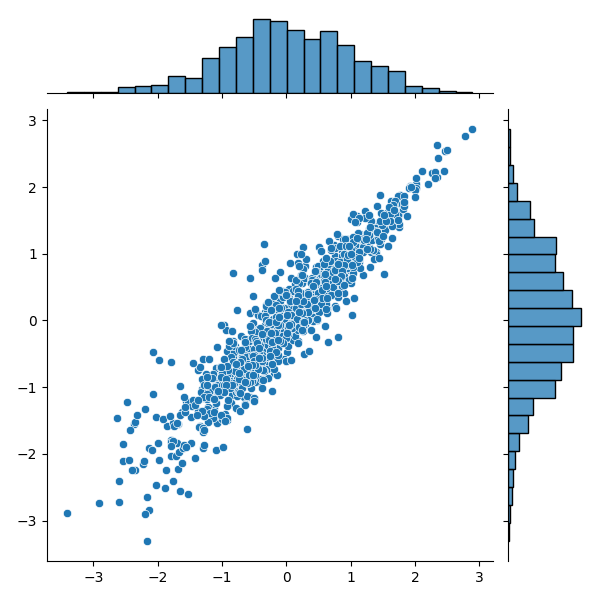}}%
\caption{Examples of non-Gaussian process whose marginals are standard normal}
\label{fig:2copulas}
\end{figure}

\subsection*{Mardia's test}
\noindent
$\bullet$ $\hat{B}_{1,i.i.d}$: Under the null hypothesis $H$, the rejection rate surpasses the nominal level. That $\hat{B}_{1,i.i.d}$ over-rejects $H$ is due to the one-dimensional marginal being time -correlated. Such observation was already formulated by \citep{Moor82:as} and \citep{Gass75:biom} who showed that the 
correlation among samples is confounded with lack of Normality.
\\
$\bullet \hat{B}_{1,i.i.d}$ and $\hat{B}_{1}$ test one-dimensional marginals only, therefore they are always conservative.
\\
\smallskip
$\bullet \hat{B}_{2}$: The rejection rates do not differ substantially from the nominal level when data is distributed according to bivariate Gaussian. Under $\bar{H}$, this test has very high rejection rates, which confirms the necessity of taking into account the full dimension to design a powerful test.
\subsection{Detection of a time-series embedded in Gaussian noise}
In this simulation, the detection of an additive corruption in a Gaussian process is considered:
\begin{equation} y(n)= x(n)+ k b(n)\end{equation}
where $x(n)$ is a first order auto-regressive process AR(1): $x(n) = 0.8x(n-1)+\eta(n)$ and where $\eta\underset{iid}{\sim}\mathcal{N}(0,S)$; $b(n)=0.8b(n-1)-0.5b(n-2)+\epsilon(n)$, where $\epsilon$ follows a double-exponential distribution with unit scale parameter. \\
We perform $500$ replications of $\{y(n)\}$ of total length  $N_{tot} = n_{drop} + N$, the first $n_{drop} = 1000$ observations at the beginning of the sample are discarded to alleviate side effects and reduce the dependence on initial values: $x(1)=\eta(1)$ and $b(1) = \epsilon(1)$. For each data record, the covariance function
$\gamma_{a,b}(i)$ is estimated once for a fixed dimension $p$ for all the test statistics. \\
Testing the normality of the process $y(n)$ can be accomplished by standard scalar tests. By exploiting the results in Section \ref{sec-embed-exps}, we propose to test the joint normality of its successive values: $\vect{x}(n)=(y(2n+1),y(2n+2))^{T}$; Note that here $\delta=2$. \\
The normality test can be reformulated in terms of the detection of an \textit{unknown} non-Gaussian signal embedded in Gaussian noise.
The ability of the test to detect the presence of $b(n)$ for different $SNR=k^{2}\frac{\E{b(n)^{2}}}{\E{x(n)^{2}}}$ is reported in Figure \ref{fig:detection}.

\begin{figure}[htbp]
    \centering
    \includegraphics[width=.5\textwidth]{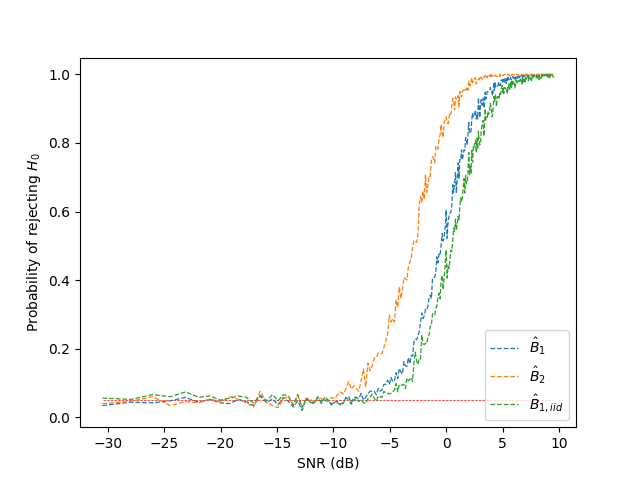}
    \caption{Empirical rejection rate at $\alpha = 5\%$ (in red dashed horizontal line) for 300 SNR values in logarithmic scale (dB)}
    \label{fig:detection}
\end{figure}

As SNR increases, statistic $\hat{B}_{2}$ is the first to detect the presence of an additive non-Gaussian process, followed by $\hat{B}_{1}$ and $\hat{B}_{1,i.i.d}$ whose behaviors do not differ substantially.

\newpage
\section{Concluding remarks}

Mardia's multivariate kurtosis, $\hat{B}_{p}$, is intended  to test the \textit{joint normality} when statistically independent realizations are available.
Without assuming the latter independence, we derive in this paper the asymptotic distribution of the multivariate kurtosis under the null hypothesis. Limited by  the length of the expressions for $p>3$, the exact expressions are reported only in the bivariate case. 

There are many ways to construct non-Gaussian processes with Gaussian marginals, as illustrated by copulas, and scalar tests often  lead to misdetections, whereas our test continues to be powerful. Our test also proves to be useful for scalar processes, for example by testing the joint normality of successive values of a time-series. 

\bigskip

\begin{paragraph}{\textbf{Acknowledgment}}
This work has been partially supported by the MIAI chair ``Environmental issues underground''of Institut MIAI@Grenoble Alpes (ANR-19-P3IA-0003).

\end{paragraph}

\bibliographystyle{elsarticle-num}
\bibliography{normality}

\section{Appendices}\label{sec:appendix}

\subsection{McCullagh's bracket notation and expression of the higher moments under the null hypothesis}\label{Appendix-A1}

 McCullagh's bracket notation \citep{Mccu87} allows to write into a compact form a sum of terms that can be deduced from each other by generating all possible partitions of the same type. For instance, we have the following expression for fourth order  moments $M_{abcd}$ of a zero-mean multivariate normal variable with covariance $\matr{S}$:
\begin{equation}
M_{abcd}=S_{ab}S_{cd}+S_{ac}S_{bd}+S_{ad}S_{bc}=[3]S_{ab}S_{cd}
\end{equation}
Moments of higher order can be found  easily:
\begin{eqnarray}
\textrm{order~6:~}\, & M_{abcdef} &= [15] S_{ab}S_{cd}S_{ef} \\
\textrm{order~8:~}\, & M_{abcdefgh} &= [105] S_{ab}S_{cd}S_{ef}S_{gh} \\
\textrm{order~10:} & M_{abcdefghij} &= [945] S_{ab}S_{cd}S_{ef}S_{gh}S_{ij} \\
\textrm{order~12:} & M_{abcdefghijk\ell} &= [10395] S_{ab}S_{cd}S_{ef}S_{gh}S_{ij}S_{k\ell} \\
\textrm{order~14:} & M_{abcdefghijk\ell mn} &= [135135] S_{ab}S_{cd}S_{ef}S_{gh}S_{ij}S_{k\ell}S_{mn} \\
\textrm{order~16:} & M_{abcdefghijk\ell mnpq} &= [2027025] S_{ab}S_{cd}S_{ef}S_{gh}S_{ij}S_{k\ell}S_{mn}S_{pq}
\end{eqnarray}
since it is well known that there are $[\frac{2r!}{2^r\,r!}]$ terms in the moment of order $2r$.

\subsection{Calculation methodology}\label{Appendix-A2}
Remind that, as introduced in Lemma \ref{lem:DLBp},   $A_{\alpha^l\beta^l}=\vect{x}(\alpha_l)^TG\vect{x}(\beta_l)$, where $G$ stands for the true precision matrix of the process whose terms are $G_{r,c}$, and where $(r,c) \in \{1,\ldots,p\}^2$. 

 Referring to the expression of $\hat B_p$ or $\hat B_p^2$ as derived from equation (\ref{eq:DLBp}), it appears that the indices $(\alpha_l, \beta_l)$ take values on a restricted set $\mathcal S=\{i, j, k, \ldots\}$, and $|\mathcal S| \ll N$.  
The following compact notation is therefore  introduced
\begin{equation}\label{eq:Annex_2}
\mu_{r_1\ldots r_L c_1\ldots c_L }^{\alpha_1\ldots\alpha_L\beta_1\ldots\beta_L}=M_{i^{\eta_i}j^{\eta_j}k^{\eta_k}....}
\end{equation}
where 
\begin{equation*}
\eta_i = \sum_{l=1}^L \left(\I_{[\alpha_l=i]}+ \I_{[\beta_l=i]}\right), \forall i \in \mathcal S
\end{equation*}
Note that  the subscripts $r_1, ...c_1...$ are skipped here for sake of readability, though any permutation of the superscripts in equation (\ref{eq:mu}) requests the corresponding permutation of the subscripts. It is easier to describe the general methodology by the typical example below. 

\subsubsection*{Example}
Consider the moment $\E{A_{nn}A_{nj}A_{jk}A_{kn}}$. According to equation (\ref{eq:Annex_1}) it will be expanded as a sum of moments of order 8 (i.e. $L=4$); using the compact notation from equation (\ref{eq:Annex_2}), we get 
\begin{eqnarray}\label{eq:Annex_3}
\E{ A_{nn} A_{nj}A_{jk}A_{kn} } =\sum_{ ((r_i,c_i)_{i=1\ldots 4})= 1 }^p G_{r_1c_1}G_{r_2c_2}G_{r_3c_3}G_{r_4c_4}
 \mu_{r_1c_1r_2c_2r_3c_3r_4c_4}^{nnnjjkkn} \nonumber \\
 = \sum_{ ((r_i,c_i)_{i=1\ldots 4})= 1 }^p G_{r_1c_1}G_{r_2c_2}G_{r_3c_3}G_{r_4c_4}M_{n^4j^2k^2} 
\end{eqnarray}
The sum involves $2^{2L}=64$ terms. It is reminded that the coefficients $r_i$ or $ c_i$ indicate the coordinate of the vector process (or space coordinate, thus taking values on $\{1,\ldots, p\}$) , whereas  time indices $n, j, k$ tale values on$\{1,\ldots, N\}$. Following McCullagh's notations, under the assumption ($H_0$) that the $p$-dimensional process is centered and jointly Gaussian, for this particular 8-th order moment 
\begin{equation*}
M_{abcdefgh}=[105]S_{ab}S_{cd}S_{ef}S_{gh}
\end{equation*}
which expresses that under $H_0$, higher even order moments (odd-order moments are zero) may be expanded as sums of products of second order moments. It must be reminded that here,  $a,b,c,d,e,f,g,h$ stand for 'meta-indices' defined  in the present example by $(n,r_1),(n,c_1),(n,r_2),(n,c_4),(j,c_2),(j,r_3),(k,c_3),(k,r_4)$ respectively, as it appears in equation (\ref{eq:Annex_3}). Plugging the above expansion in equation (\ref{eq:Annex_3}) leads to summing over $64\times 105$ terms! However, in most cases of interest many terms may be grouped together and highlight the behavior of equation (\ref{eq:Annex_1}). The case $p=1$ is briefly sketched below as an illustration. 

The case $p=1$ implies that $r_i=c_i=1$ $\forall i \in \{1,\ldots, (L=4)\}$; the particular 8-th order moment in equation (\ref{eq:Annex_3}) may be simply written as $M_{n^4j^2k^2}$, whose expansion into sum of products of second order moments will involve the following products : 
(as there is no ambiguity in this case,  we set $M_{ij}\overset{nota.}{=}S_{ij}$), 
\begin{eqnarray*}
S_{nn}S_{nn}S_{jj}S_{kk} & \mbox{appearing~}  3  \mbox{~times} \\
S_{nn}S_{nn}S_{jk}S_{jk} & \mbox{appearing~}  6  \mbox{~times}\\
S_{nn}S_{nj}S_{nj}S_{kk} & \mbox{appearing~} 12  \mbox{~times}\\
S_{nk}S_{nk}S_{nj}S_{nj} & \mbox{appearing~} 24  \mbox{~times}\\
S_{nj}S_{jk}S_{nk}S_{nn} & \mbox{appearing~}  48  \mbox{~times}\\
S_{nn}S_{nk}S_{nk}S_{jj} & \mbox{appearing~}  12  \mbox{~times}
\end{eqnarray*}

For example the number of occurences of the term of type  $S_{nk}S_{nk}S_{nj}S_{nj}$  is given by 
$$
(4\times2 \times 3\times1)/2\times (2\times2\times1\times1)/2 = 24
$$ 
where $4\times 2$ stand for the number of possible choices for index $i$ (one out of 4) times the number of possible choices for index $k$ (one out of 2); then  $3\times 1$ stand for the number of remaining possibilities to select index $i$ times the remaining choices for $k$; Division by 2 accounts for the fact that permutations of terms  $S_{ik}$  were counted twice. All other occurence calculations follow the same guidelines.   Finally, one gets for the case $p=1$
\begin{eqnarray*}
 M_{n^{4}jjkk} = 3S_{nn}^2 S_{jj}S_{kk}+6S_{nn}^{2}S_{jk}S_{jk} + 12S_{nn}S_{ij}^{2}S_{kk}+24S_{nk}^{2}S_{nj}^{2}+.... \\
48S_{nj}S_{jk}S_{nk}S_{nn}+12S_{nn}S_{nk}^{2}S_{jj}
 \end{eqnarray*}
 which can be directly plugged into equation (\ref{eq:Annex_3}).
 Note that the sum of all coefficient is actually 105, as expected for an 8-th order moment.

The cases $p\geq 2$ turns out to be a bit more complicated, as one has to deal with the 'meta-indices' directly. However counting the number of configurations involving the same time indices follows the same lines as in the case $p=1$. Going back to the example introduced  above  for $p=2$, one gets 

\begin{eqnarray*}
\E{ A_{nn} A_{nj}A_{jk}A_{kn} } =\sum_{ ((r_i,c_i)_{i=1\ldots 4})= 1 }^p G_{r_1c_1}G_{r_2c_2}G_{r_3c_3}G_{r_4c_4}
\left\{ \right.
\left[ 3\right] \mu_{r_1c_1}^{nn} \mu_{r_2c_4}^{nn} \mu_{c_2r_3}^{jj} \mu_{c_3r_4}^{kk}+ \\
\left[ 6\right] \mu_{r_1c_1}^{nn} \mu_{r_2c_4}^{nn}  \mu_{c_2c_3}^{jk} \mu_{r_3r_4}^{jk} +  
\left[ 12\right] \mu_{r_1c_1}^{nn} \mu_{r_2c_2}^{nj} \mu_{c_4r_3}^{nj} \mu_{c_3r_4}^{kk} + 
\left[ 24\right] \mu_{r_1c_3}^{nk} \mu_{c_1r_4}^{nk}\mu_{r_2c_2}^{nj} \mu_{c_4r_3}^{nj} + \\
\left[ 48\right] \mu_{r_1c_2}^{nj} \mu_{r_3c_3}^{jk} \mu_{c_1r_4}^{nk} \mu_{r_2c_4}^{nn} + 
\left[12\right] \mu_{r_1c_1}^{nn} \mu_{r_2c_3}^{nk} \mu_{c_4r_4}^{nk} \mu_{c_2r_3}^{jj}
\left. \right\}
\end{eqnarray*}

where we have used notations $\mu_{rc}^{\alpha\beta}$  to emphasize that the permutations (whose number is indicated using McCullagh's brakets) are applied on the 'meta-indices' and grouped such that they share the same 'time structure'; This allow to get the same values as in the case $p=1$, though replacing the scalar coefficients by McCullagh's brakets. 

\subsection{Multivariate moments up to order 12}\label{app:moments}
In this section, we give all moments of a zero-mean multivariate normal variable of even order. Most of these expressions have not been reported in the literature. 
In addition, for the sake of readability, when an index is repeated more than three times, we assume an alternative notation, for instance at order 10:
$$
M_{iiiiijjjjk} = M_{i^5j^4k}
$$
Furthermore, we use notation introduced in (\ref{eq:Annex_2}) involving meta-indices; more precisely, since each subscript is always associated with a superscript, we may omit the subscript. In order to lighten notation, especially when terms need to be raised to a power, we put the latter superscript in subscript. For instance in (\ref{Miiij-eq}), $M_{abcd}^{iiij}$ is replaced by $M_{iiij}$. In the list below, moments are sorted by increasing $D$, where $D$ denotes the number of distinct indices.

\medskip
\noindent\textbf{Order 4, D=2.}
\begin{eqnarray}
M_{iiij} &=& [3] \mu_{ab}^{ii} \mu_{cd}^{ij} \label{Miiij-eq} \\
M_{iijj} &=& [2] \mu_{ab}^{ij}\mu_{cd}^{ij} +  \mu_{ab}^{ii} \mu_{cd}^{jj}
\end{eqnarray}
\textbf{Order 4, D=3.}
\begin{eqnarray}
M_{iijk} &=& \mu_{ab}^{ii}\mu_{cd}^{jk} + [2] \mu_{ab}^{ij}\mu_{cd}^{ik}
\end{eqnarray}
\noindent\textbf{Order 6, D=2.}
\begin{eqnarray}
M_{i^5j} &=& [15] \mu_{ab}^{ii}\mu_{cd}^{ii} \mu_{ef}^{ij}\\
M_{i^4jj} &=&  [12] \mu_{ae}^{ij}\mu_{bf}^{ij} \mu_{cd}^{ii} + [3] \mu_{ab}^{ii}\mu_{cd}^{ii} \mu_{ef}^{jj}\\
M_{iiijjj} &=& [6] \mu_{ad}^{ij}\mu_{be}^{ij}\mu_{df}^{ij} + [9] \mu_{ab}^{ii} \mu_{cd}^{ij}\mu_{ef}^{jj}
\end{eqnarray}
\noindent\textbf{Order 6, D=3.}
\begin{eqnarray}
M_{i^4jk} &=&[3] \mu_{ab}^{ii}\mu_{cd}^{ii} \mu_{ef}^{jk} + [12] \mu_{ae}^{ij}\mu_{bf}^{ik}\mu_{cd}^{ii}\\
M_{iiijjk} &=& [6] \mu_{ad}^{ij}\mu_{be}^{ij} \mu_{cf}^{ik} + [6] \mu_{ad}^{ij}\mu_{bc}^{ii}\mu_{ef}^{jk} + [3] \mu_{ab}^{ii} \mu_{de}^{jj} \mu_{cf}^{ik} \\
M_{iijjkk} &=& \mu_{ab}^{ii} \mu_{cd}^{jj} \mu_{ef}^{kk} + [2] \mu_{ab}^{ii}\mu_{ce}^{jk}\mu_{bf}^{jk} + [2] \mu_{cd}^{jj} \mu_{ae}^{ik}\mu_{bf}^{ik}+ [2] \mu_{ef}^{kk}\mu_{ac}^{ij}\mu_{bd}^{ij} \nonumber \\&&+ [8] \mu_{ac}^{ij} \mu_{de}^{jk} \mu_{bf}^{ik} 
\end{eqnarray}

\noindent\textbf{Order 8, D=2.}
\begin{eqnarray}
M_{i^7j} &=& [105] \mu_{ab}^{ii}\mu_{cd}^{ii}\mu_{ef}^{ii} \mu_{gh}^{ij}\\
M_{i^6jj} &=&  [90] \mu_{ag}^{ij}\mu_{bh}^{ij} \mu_{cd}^{ii}\mu_{ef}^{ii} +  [15] \mu_{ab}^{ii}\mu_{cd}^{ii}\mu_{ef}^{ii} \mu_{gh}^{jj}\\
M_{i^5jjj} &=& [60] \mu_{af}^{ij}\mu_{bg}^{ij}\mu_{ch}^{ij} \mu_{de}^{ii}+ [45] \mu_{ab}^{ii}\mu_{cd}^{ii}\mu_{ef}^{ij} \mu_{gh}^{jj} \\
M_{i^4j^4} &=& [9] \mu_{ab}^{ii}\mu_{cd}^{ii} \mu_{ef}^{jj}\mu_{gh}^{jj} + [72] \mu_{ab}^{ii} \mu_{ce}^{ij}\mu_{df}^{ij} \mu_{gh}^{jj} + [24] \mu_{ae}^{ij}\mu_{bf}^{ij}\mu_{cg}^{ij}\mu_{dh}^{ij} 
\end{eqnarray}
\noindent\textbf{Order 8, D=3.}
\begin{eqnarray}
M_{i^6jk} &=& [15] \mu_{ab}^{ii}\mu_{cd}^{ii}\mu_{ef}^{ii}\mu_{gh}^{jk} + [90] \mu_{ab}^{ii}\mu_{cd}^{ii} \mu_{eg}^{ij} \mu_{fg}^{ik} \\
M_{i^5jjk} &=& [30] \mu_{ab}^{ii}\mu_{cd}^{ii} \mu_{ef}^{ij} \mu_{gh}^{jk} + [60]\mu_{af}^{ij} \mu_{bc}^{ii}\mu_{dh}^{ik} + [15] \mu_{ab}^{ii}\mu_{cd}^{ii}\mu_{fg}^{jj} \mu_{eh}^{ik}\\
M_{i^4jjjk} &=& [9] \mu_{ab}^{ii}\mu_{cd}^{ii} \mu_{ef}^{jj}\mu_{gh}^{jk} + [36] \mu_{ef}^{jj} \mu_{ag}^{ij} \mu_{bh}^{ik} \mu_{cd}^{ii}\nonumber\\&&+[24] \mu_{ae}^{ij}\mu_{bf}^{ij}\mu_{cg}^{ij} \mu_{dh}^{ik} + [36] \mu_{ab}^{ii} \mu_{ce}^{ij}\mu_{df}^{ij} \mu_{gh}^{jk} \\
M_{i^4jjkk} &=& [3] \mu_{ab}^{ii}\mu_{cd}^{ii} \mu_{ef}^{jj} \mu_{gh}^{kk} + [6] \mu_{ab}^{ii}\mu_{cd}^{ii} \mu_{eg}^{jk}\mu_{fh}^{jk} + [12] \mu_{ab}^{ii} \mu_{ce}^{ij}\mu_{df}^{ij} \mu_{gh}^{kk}\nonumber \\ &&+ [24] \mu_{ag}^{ik}\mu_{bg}^{ik} \mu_{ce}^{ij}\mu_{df}^{ij} + [48] \mu_{ae}^{ij} \mu_{fg}^{jk} \mu_{bh}^{ik} \mu_{cd}^{ii} \nonumber\\&&+ [12] \mu_{ab}^{ii}\mu_{bg}^{ik}\mu_{ch}^{ik}\mu_{ef}^{jj}\\
M_{iiijjjkk} &=& [9] \mu_{ab}^{ii} \mu_{cd}^{ij} \mu_{ef}^{jj} \mu_{gh}^{kk} + [18] \mu_{ab}^{ii} \mu_{cd}^{ij} \mu_{eg}^{jk}\mu_{fh}^{jk} + [6] \mu_{ad}^{ij}\mu_{be}^{ij}\mu_{cf}^{ij} \mu_{gh}^{kk} \nonumber \\&& + [18] \mu_{ag}^{ik}\mu_{bh}^{ik} \mu_{cd}^{ij} \mu_{ef}^{jj} + [36] \mu_{ad}^{ij}\mu_{be}^{ij} \mu_{cg}^{ik} \mu_{fh}^{jk} \nonumber\\&& + [18] \mu_{ag}^{ik} \mu_{dh}^{jk} \mu_{bc}^{ii} \mu_{ef}^{jj}
\end{eqnarray}
\noindent\textbf{Order 10, D=2.}
\begin{eqnarray}
M_{i^{9}j} &=& [945] \mu_{ab}^{ii}\mu_{cd}^{ii}\mu_{ef}^{ii}\mu_{gh}^{ii} \mu_{m\ell}^{ij} \\
M_{i^{8}jj}&=&[105] \mu_{ab}^{ii}\mu_{cd}^{ii}\mu_{ef}^{ii}\mu_{gh}^{ii} \mu_{m\ell}^{jj} + [840] \mu_{ab}^{ii}\mu_{cd}^{ii}\mu_{ef}^{ii} \mu_{gm}^{ij}\mu_{h\ell}^{ij} \\
M_{i^{7}jjj}&=& [315] \mu_{ab}^{ii}\mu_{cd}^{ii}\mu_{ef}^{ii} \mu_{gh}^{ij} \mu_{m\ell}^{jj}+ [630] \mu_{ah}^{ij}\mu_{bm}^{ij}\mu_{c\ell}^{ij} \mu_{de}^{ii}\mu_{fg}^{ii} \\
M_{i^{6}j^{4}} &=& [45] \mu_{ab}^{ii}\mu_{cd}^{ii}\mu_{ef}^{ii} \mu_{gh}^{jj}\mu_{m\ell}^{jj} + [360] \mu_{ag}^{ij}\mu_{bh}^{ij}\mu_{cm}^{ij}\mu_{d\ell}^{ij} \mu_{ef}^{ii} \nonumber\\&&+ [540] \mu_{ab}^{ii}\mu_{cd}^{ii} \mu_{eg}^{ij}\mu_{fh}^{ij} \mu_{ml}^{jj}\\
M_{i^{5}j^{5}} &=& [120] \mu_{af}^{ij}\mu_{bg}^{ij}\mu_{ch}^{ij}\mu_{dm}^{ij}\mu_{e\ell}^{ij} \nonumber\\&&+ [225] \mu_{ab}^{ii}\mu_{cd}^{ii} \mu_{ef}^{ij} \mu_{gh}^{jj}\mu_{m\ell}^{jj} + [600] \mu_{fg}^{jj} \mu_{ah}^{ij}\mu_{bm}^{ij}\mu_{c\ell}^{ij}\mu_{de}^{ii} 
\end{eqnarray}

\noindent\textbf{Order 10, D=3.}
\begin{eqnarray}
M_{i^{8}jk}&=& [105] \mu_{ab}^{ii}\mu_{cd}^{ii}\mu_{ef}^{ii}\mu_{gh}^{ii} \mu_{m\ell}^{jk} + [840] \mu_{ab}^{ii}\mu_{cd}^{ii}\mu_{ef}^{ii} \mu_{gm}^{ij} \mu_{h\ell}^{ik}\\
M_{i^{7}jjk}&=&  [210] \mu_{ab}^{ii}\mu_{cd}^{ii}\mu_{ef}^{ii}\mu_{gh}^{ij} \mu_{m\ell}^{jk} + [630] \mu_{ab}^{ii}\mu_{cd}^{ii} \mu_{eh}^{ij}\mu_{fm}^{ij} \mu_{h\ell}^{ik} \nonumber\\&&+[105] \mu_{ab}^{ii}\mu_{cd}^{ii}\mu_{ef}^{ii} \mu_{hm}^{jj}\mu_{g\ell}^{ik} \\
M_{i^{6}jjjk}&=&  [45] \mu_{ab}^{ii}\mu_{cd}^{ii}\mu_{ef}^{ii} \mu_{gh}^{jj} \mu_{m\ell}^{jk} + [270] \mu_{ab}^{ii}\mu_{cd}^{ii} \mu_{eg}^{ij}\mu_{fh}^{ij} \mu_{m\ell}^{jk} \nonumber\\&&+ [360] \mu_{ag}^{ij}\mu_{bh}^{ij}\mu_{cm}^{ij} \mu_{d\ell}^{ik} \mu_{ef}^{ii} + [270] \mu_{a\ell}^{ik} \mu_{gh}^{jj} \mu_{b\ell}^{ij} \mu_{cd}^{ii}\mu_{ef}^{ii} \\
M_{i^{6}jjkk}&=& [15] \mu_{ab}^{ii}\mu_{cd}^{ii}\mu_{ef}^{ii} \mu_{gh}{jj} \mu_{m\ell}^{kk} + [30] \mu_{ab}^{ii}\mu_{cd}^{ii}\mu_{ef}^{ii} \mu_{gm}^{jk}\mu_{h\ell}^{jk} \nonumber\\&&+ [90] \mu_{ab}^{ii}\mu_{cd}^{ii} \mu_{eg}^{ij}\mu_{fh}^{ij} \mu_{m\ell}^{kk} + [90] \mu_{ab}^{ii}\mu_{cd}^{ii} \mu_{em}^{ik}\mu_{f\ell}^{ik} \mu_{gh}^{jj} \nonumber\\&&+ [360] \mu_{ab}^{ii} \mu_{cg}^{ij}\mu_{dh}^{ij} \mu_{em}^{ik}\mu_{f\ell}^{ik} + [360] \mu_{ab}^{ii}\mu_{cd}^{ii}\mu_{eg}^{ij} \mu_{fm}^{ik} \mu_{h\ell}^{jk} \\
M_{i^{5}j^{4}k}&=&  [45] \mu_{ab}^{ii}\mu_{cd}^{ii} \mu_{fg}^{jj}\mu_{hm}^{jj}\mu_{e\ell}^{ik} + [360] \mu_{ab}^{ii} \mu_{cf}^{ij}\mu_{dg} \mu_{hm}^{jj} \mu_{e\ell}^{ik} \nonumber\\&&+ [120] \mu_{ag}^{ij}\mu_{bf}^{ij}\mu_{cg}^{ij}\mu_{dh}^{ij}\mu_{e\ell}^{ik} + [180] \mu_{ab}^{ii}\mu_{cd}^{ii} \mu_{ef}^{ij} \mu_{gh}^{jj} \mu_{m\ell}^{jk} \nonumber \\&&+ [240] \mu_{ab}^{ii} \mu_{cf}^{ij}\mu_{dg}^{ij}\mu_{eh}^{ij} \mu_{m\ell}^{jk} \\
M_{i^{5}jjjkk}&=& [45] \mu_{ab}^{ii}\mu_{cd}^{ii} \mu_{ef}^{ij} \mu_{gh}^{jj} \mu_{m\ell}{kk} + [60] \mu_{ab}^{ii}\mu_{cf}^{ij}\mu_{dg}^{ij}\mu_{eh}^{ij} \mu_{m\ell}^{kk} \nonumber\\&&+ [90] \mu_{ab}^{ii}\mu_{cd}^{ii}\mu_{ef}^{ij} \mu_{gm}^{jk}\mu_{h\ell}^{jk} + 360 \mu_{ab}^{ii} \mu_{cf}^{ij}\mu_{dg}^{ij}\mu_{em}^{ik} \mu_{h\ell}^{jk} \nonumber\\&&+ [90] \mu_{ab}^{ii} \mu_{fg}^{jj}\mu_{cm}^{ik}\mu_{h\ell}^{jk} + [180] \mu_{ab}^{ii} \mu_{cf}^{ij} \mu_{gh}^{jj} \mu_{dm}^{ik}\mu_{e\ell}^{ik}\nonumber\\&& + [120] \mu_{af}^{ij}\mu_{bg}^{ij}\mu_{ch}^{ij} \mu_{dm}^{ik}\mu_{e\ell}^{ik}\\
M_{i^{4}jjjkkk}&=& [27] \mu_{ab}^{ii}\mu_{cd}^{ii} \mu_{ef}^{jj} \mu_{gh}^{jk} \mu_{m\ell}^{kk} + [18] \mu_{ab}^{ii}\mu_{cd}^{ii}\mu_{eh}^{jk}\mu_{fm}^{jk}\mu_{g\ell}^{ij} \nonumber\\&&+ [108] \mu_{ab}^{ii} \mu_{ce}^{ij}\mu_{df}^{ij} \mu_{gh}^{jk}\mu_{m\ell}^{kk} + [108] \mu_{ab}^{ii}\mu_{ef}^{jj} \mu_{gh}^{jk}\mu_{cm}^{ik}\mu_{d\ell}^{ik} \nonumber \\&&+ [108] \mu_{ab}^{ii} \mu_{ef}^{jj} \mu_{cg}^{ij} \mu_{dh}^{ik} \mu_{m\ell}^{kk} + [216] \mu_{ab}^{ii} \mu_{ce}^{ij} \mu_{dh}^{ik} \mu_{fm}^{jk}\mu_{g\ell}^{jk}  \nonumber \\&&+ [72] \mu_{ae}^{ij}\mu_{bf}^{ij}\mu_{cg}^{ij} \mu_{dh}^{ik} \mu_{m\ell}^{kk} + [216] \mu_{ah}^{ik}\mu_{bm}^{ik} \mu_{ce}^{ij}\mu_{df}^{ij}\mu_{g\ell}^{jk}\nonumber\\&& + [72] \mu_{ah}^{ik}\mu_{bm}^{ik}\mu_{c\ell}^{ik} \mu_{de}^{ij} \mu_{fg}^{jj}\\
M_{i^{4}j^{4}kk}&=& [9]\mu_{ab}^{ii}\mu_{cd}^{ii} \mu_{ef}^{jj}\mu_{gh}^{jj} \mu_{m\ell}^{kk} + [72] \mu_{ab}^{ii} \mu_{ce}^{ij}\mu_{df}^{ij} \mu_{gh}^{jj} \mu_{m\ell}^{kk} \nonumber\\&&+ [24] \mu_{ae}^{ij}\mu_{bf}^{ij}\mu_{cg}^{ij}\mu_{dh}^{ij} \mu_{m\ell}^{kk} + [36] \mu_{ab}^{ii} \mu_{ef}^{jj} \mu_{gm}^{jk}\mu_{h\ell}^{jk}  \nonumber \\&& + [144] \mu_{ab}^{ii} \mu_{ce}^{ij}\mu_{df}^{ij} \mu_{gm}^{jk}\mu_{h\ell}^{jk} + [36] \mu_{ab}^{ii} \mu_{ef}^{jj}\mu_{gh}^{jj} \mu_{cm}^{ik}\mu_{d\ell}^{ik} \nonumber\\&& + [144] \mu_{ae}^{ij}\mu_{bf}^{ij} \mu_{gh}^{jj} \mu_{cm}^{ik}\mu_{d\ell}^{ik} + [288] \mu_{ab}^{ii} \mu_{ef}^{jj} \mu_{cg}^{ij} \mu_{dm}^{ik} \mu_{h\ell}^{jk}  \nonumber \\&&+ [192] \mu_{ae}^{ij}\mu_{bf}^{ij}\mu_{cg}^{ij}\mu_{dm}^{ik}\mu_{h\ell}^{jk}
\end{eqnarray}

\bigskip

\subsection{Particular results when $p=1$}
Here we remind that $\mu_{11}^{ij} = S_{ij}$.

\noindent\textbf{Order 12, p=1, D=2.}
\begin{eqnarray}
M_{i^{11}j}&=& 10395 S_{ii}^{5}S_{ij} +  9450 S_{ii}^4 S_{ij}^2 \\ 
M_{i^{9}jjj}&=& 2835 S_{ii}^4 S_{ij} S_{jj}+ 7560 S_{ii}^3 S_{ij}^3 \\
M_{i^{8}j{4}} &=& 5040 S_{ij}^{4} S_{ii}^2 + 315 S_{ii}^4 S_{jj}^2 + 5040 S_{ii}^3 S_{ij}^2 S_{jj} \\
M_{i^{7}j^{5}} &=& 1575 S_{ii}^{3} S_{ij} S_{jj}^2 + 6300 S_{ii}^2 S_{ij}^3 S_{jj} + 2520 S_{ii} S_{ij}^5 \\
M_{i^{6}j^{6}} &=& 720 S_{ij}^6 + 225 S_{ii}^{3} S_{jj}^3 + 5400 S_{ii} S_{ij}^4 S_{jj} + 4050 S_{ii}^2 S_{ij}^2 S_{jj}^2
\end{eqnarray}
\noindent\textbf{Order 12, p=1, D=3.}
\begin{eqnarray}
M_{i^{10}jk}&=& 945 S_{ii}^5 S_{jk} + 9450 S_{ik} S_{ij} S_{ii}^4\\
M_{ijjk}&=& 945 S_{ii}^4 S_{jj} S_{ik} + 7560 S_{ii}^3 S_{ij}^2 S_{ik} + 1890 S_{ii}^4 S_{ij} S_{jk}\\
M_{i^{8}jjjk}&=& 315 S_{ii}^4 S_{jj} S_{jk} + 2520 S_{ii}^3 S_{ij} S_{jj} S_{ik} + 2520 S_{ii}^3 S_{ij}^2 S_{jk} \nonumber\\&&+ 5040 S_{ii}^2 S_{ij}^3 S_{ik}\\
M_{i^{7}j^{4}k}&=& 315 S_{ii}^3 S_{jj}^2 S_{ik} + 3780 S_{ii}^2 S_{ij}^2 S_{ik} + 1260 S_{ii} S_{ij}^4 S_{ik} + 1260 S_{ii}^3 S_{ij} S_{jj} S_{jk}\nonumber\\&& + 3780 S_{ii}^2 S_{ij}^3 S_{jk}\\
M_{i^{8}jjkk} &=& 105 S_{ii}^4 S_{jj} S_{kk} + 210 S_{ii}^4 S_{jk}^2 + 840 S_{ii}^3 S_{ij}^2 S_{kk} + 840 S_{ii}^3 S_{ik}^2 S_{jj}\nonumber\\&&+ 5040 S_{ii}^2 S_{ij}^2 S_{ik}^2 + 3360 S_{ii}^3 S_{ik} S_{ij} S_{jk}
\end{eqnarray}
\noindent\textbf{Order 12, p=1, D=4.}
\begin{eqnarray}
M_{i^{4}j^{4}kk\ell \ell}&=&3 S_{ii}^{2}[3 S_{jj}^{2} S_{kk}S_{\ell \ell} + 6 S_{jj}^2 S_{k\ell}^{2} + 12 S_{jj}S_{jk}^{2}S_{\ell \ell} + 24 S_{j\ell}^{2}S_{jk}^{2}\nonumber\\&& + 48 S_{jk} S_{k\ell} S_{j\ell} S_{jj} + 12 S_{jj} S_{j\ell}^{2} S_{kk}] + 3 S_{jj}^{2}[12 S_{ii}S_{ik}^{2}S_{\ell \ell} \nonumber\\&&+ 24 S_{il}^{2}S_{ik}^2 + 48 S_{ik}S_{k\ell}S_{i\ell}S_{ii} + 12 S_{ii}S_{i\ell}^{2} S_{kk}]\nonumber \\&& + 24 S_{ij}^{4} S_{kk} S_{\ell \ell} + 48 S_{ij}^{4} S_{k\ell}^{2} + 96 S_{ij}^{3}[2 S_{ik}S_{jk}S_{\ell \ell} + 2 S_{i\ell} S_{j\ell} S_{kk} \nonumber \\ && + 4 S_{ik} S_{j\ell}S_{k\ell} + 4 S_{i\ell} S_{jk} S_{\ell k}] + 72 S_{ij}
       ^{2}[ 4 S_{ik}^{2} S_{j\ell}^{2} + 4 S_{jk}^{2} S_{i\ell}^{2} \nonumber\\&&+ 16 S_{ik} S_{i\ell} S_{jk} S_{j\ell} + S_{ii}S_{jj}S_{kk}S_{\ell \ell} + 2 S_{ii} S_{jj} S_{k\ell}^{2}] + 12 S_{ik}^{2}[12 S_{ji}^{2}\nonumber \\&& \times S_{jj} S_{\ell \ell} + 48 S_{ij} S_{i\ell} S_{j\ell} S_{jj} + 12 S_{jj}S_{j\ell}^{2}S_{ii}+ 12 S_{i\ell}^{2} [12 S_{ji}^{2} S_{jj} S_{kk} \nonumber\\&& + 48 S_{ij}S_{ik}S_{jk}S_{jj} + 12 S_{jj}S_{jk}^{2}S_{ii}] + 12 S_{j\ell}^{2} [12 S_{ij}^{2} S_{ii} S_{kk} \nonumber\\&& + 48 S_{ij} S_{jk} S_{ik}S_{ii}] + 12 S_{jk}^{2} [12 S_{ij}^{2} S_{ii} S_{\ell \ell}  + 48 S_{ij} S_{j\ell} S_{i\ell} S_{ii}] \nonumber\\&&+ 576 S_{ii}[S_{ik}S_{i\ell}S_{jj}S_{jk}S_{j\ell} +S_{ik}S_{ij}S_{jj}S_{j\ell}S_{k\ell}\nonumber \\&&+ S_{ik}S_{ij}S_{jj}S_{jk}S_{\ell \ell} + S_{il}S_{ij}S_{jj}S_{jk}S_{\ell k} \nonumber \\ && + S_{i\ell}S_{ij}S_{jj}S_{j\ell}S_{kk}+S_{ik}S_{ij}S_{jj}S_{\ell k}S_{j\ell}]
\end{eqnarray}
\FIN
\subsection{Computation of the mean of $\hat{B}_p(N)$}\label{app:mean}
The first step is to unfold McCullagh's bracket notation to have the explicit summation terms. For instance: 
\begin{eqnarray}
\E{A_{nn}^2} &=& \sum_{a,b=1}^{p}\sum_{c,d=1}^{p} G_{ab}G_{cd}(S_{ab}S_{cd}+ S_{ac}S_{bd} + S_{ad}S_{bc}) 
\end{eqnarray}

\medskip
\textbf{For p=1.}
\begin{eqnarray}
\E{A_{nn}^{2}} &=& 3\\
    \E{A_{nn}A_{ni}^{2}} &=& 3 + 12 \frac{S(n-i)^{2}}{S^{2}} \\
    \E{A_{ni}^{2}A_{nj}^{2}} &=& 3 + 6\frac{S(i-j)^{2}}{S^2}+12\frac{S(n-i)^{2}}{S^2}+12\frac{S(n-j)^{2}}{S^2} \nonumber\\&&+ 24 \frac{S(n-i)^{2}S_{11}(n-j)^{2}}{S^4} + 48 \frac{S(n-i)S(n-j)S(i-j)}{S^3} \nonumber \\
    \E{A_{nn}A_{nj}A_{jk}A_{kn}} &=& 3 + 6\frac{S(j-k)^{2}}{S^2} + 12 \frac{S(n-k)^{2}}{S^2} + 12 \frac{S(n-j)^{2}}{S^2} \nonumber \\&&+ 24 \frac{S(n-k)^{2}S(n-j)^{2}}{S^4} + 48 \frac{S(n-k)S(j-k)S(n-j)}{S^3} \nonumber\\
\end{eqnarray} 
The exact computation of $\E{\hat{B}_{1}}$ yields the following result:
\begin{eqnarray}
    \E{\hat{B}_{1}} &=& 3 - \frac{6}{N^2}\sum_{n,i}\frac{S(n-i)^2}{S^2} + \frac{72}{N^3} \sum_{n,i,j} \frac{S(n-j)^2 S(n-i)^2}{S^4} \nonumber \\&& +\frac{144}{N^3} \sum_{n,i,j} \frac{S(n-i)S(i-j)S(n-j)}{S^3}
\end{eqnarray}
Based on the results in \cite[p.~346-347]{Cram46}, it can be shown that $\frac{1}{N^3} \sum_{n,i,j} \frac{S(n-j)^2 S(n-i)^2}{S^4}$ and $\frac{1}{N^3} \sum_{n,i,j} \frac{S(n-i) S(i-j) S(n-j)}{S^3}$ will contribute quantities of order lower than $N^{-1}$.

\medskip
\textbf{For p=2.}
\begin{eqnarray}
\E{A_{nn}^{2}} &=& 8 \\
\E{A_{nn}A_{ni}^{2}}&=& 8 + \frac{1}{(S_{11}S_{22}-S_{12}^{2})^{2}}\Big[S_{11}S_{22}\big[2S_{12}^{2}(n-i)+2S_{12}^{2}(n-i) \nonumber \\&&+12S_{21}(n-i)S_{12}(n-i)+ 12S_{22}(n-i)S_{11}(n-i)\big]\nonumber\\&&+ S_{12}^{2}\big[ 12 S_{12}^{2}(n-i)+12S_{21}(n-i)S_{12}(n-i)\nonumber\\&&+12S_{21}^{2}(n-i)+16S_{11}(n-i)S_{22}(n-i)\big] \nonumber\\&&-28S_{12}S_{11}\big[S_{22}(n-i)S_{12}(n-i)+S_{22}(n-i)S_{12}(n-i)\big]\nonumber\\&&-28S_{12}S_{22}\big[S_{11}(n-i)S_{12}(n-i)+S_{11}(n-i)S_{12}(n-i)\big]\nonumber\\&&+ 14S_{11}^{2}S_{22}^{2}(n-i)+14S_{22}^{2}S_{11}^{2}(n-i)\Big] \end{eqnarray}

\textit{Bivariate embedding.}
\begin{eqnarray}
\E{A_{nn}A_{ni}^2} &=& 8 + \frac{1}{(C_{0}^{2}-C_{1}^{2})^2} \times [C_{0}^2\times [2(\gamma_{1}(n-i) + \gamma_{-1}(n-i))^2 \nonumber\\&&+ 8 \gamma_{-1}(n-i) \gamma_{1}(n-i)  + 40 \gamma_{0}(n-i)^2] + C_{1}^2\nonumber\\&&\times [12 (\gamma_{1}(n-i) + \gamma_{-1}(n-i))^2 -8\gamma_{-1}(n-i)\gamma_{1}(n-i) + 16 \gamma_{0}(n-i)^{2}]\nonumber\\&& - C_{0}C_{1} \times [56 \gamma_{0}(n-i) ( \gamma_{1}(n-i) + \gamma_{-1}(n-i) )]]\\
\E{A_{ni}^{2}A_{nj}^{2}} &=& 8 + \frac{1}{(C_{0}^{2}-C_{1}^{2})^2} \times [C_{0}^2\times [2 (\gamma_{1}(i-j) + \gamma_{-1}(i-j))^2\nonumber\\&&+ 2 (\gamma_{1}(n-i) + \gamma_{-1}(n-i))^2 + 2 (\gamma_{1}(n-j) + \gamma_{-1}(n-j))^2 \nonumber\\&&+ 8 \gamma_{-1}(n-i) \gamma_{1}(n-i) +  8 \gamma_{-1}(n-j) \gamma_{1}(n-j)  + 16 \gamma_{0}(i-j)^2 \nonumber\\&&+ 40 \gamma_{0}(n-j)^2+ 40 \gamma_{0}(n-i)^2] + C_{1}^2\times [4 (\gamma_{1}(i-j) + \gamma_{-1}(i-j))^2 \nonumber\\&&+ 12 (\gamma_{1}(n-j) + \gamma_{-1}(n-j))^2 + 12 (\gamma_{1}(n-i) + \gamma_{-1}(n-i))^2 \nonumber\\&&- 8 \gamma_{-1}(n-i)\gamma_{1}(n-i) - 8 \gamma_{-1}(n-j)\gamma_{1}(n-j) + 8 \gamma_{0}(i-j)
^2\nonumber\\&&+ 16 \gamma_{0}(n-j)^2 +16 \gamma_{0}(n-i)^2] - C_{0}C_{1} \times [24 \gamma_{0}(i-j)\nonumber\\&&\times(\gamma_{1}(i-j)+\gamma_{-1}(i-j)) + 56 \gamma_{0}(n-i)(\gamma_{1}(n-i)+\gamma_{-1}(n-i)) \nonumber\\&&+ 56 \gamma_{0}(n-j)(\gamma_{1}(n-j)+\gamma_{-1}(n-j)) ]
\end{eqnarray}

Following the same pattern as the mean, but with more moments involved, the computation of the variance may be conducted [\cite{Elbo21}]. 

\subsection{Sklar's theorem}\label{Sklar}
\begin{theorem}(Sklar's theorem 1959)
\begin{eqnarray}
F_{X_1,X_2}(x_1,x_2) = Pr(X_1\leq x_1, X_2 \leq x_2) = \mathcal{C}(F(x_1),G(x_2))
\end{eqnarray}
where $F_{X_1,X_2}$ is the joint cumulative distribution function (cdf) of $(X_1,X_2)$, and $F$ (resp. $G$) is the  cdf of $X_1$ (resp. $X_2$). 
If $F$, $G$ are continuous, then $\mathcal{C}$ is unique, and is defined by:
\begin{equation}
    \mathcal{C}(u_1,u_2) = F_{X_1,X_2}(F^{-1}(u_1),G^{-1}(u_2)).
\end{equation}
\end{theorem}

\end{document}